\documentclass[manuscript,nonacm,authorversion,natbib=false,hypertexnames=false]{acmart}
\usepackage[utf8]{inputenc}
\usepackage{microtype}
\usepackage{mathtools}   %
\usepackage{mathrsfs,amsthm}
\usepackage{xspace}
\usepackage{floatpag}
\usepackage{subfig}
\usepackage{tikz}
\usetikzlibrary {arrows.meta}

\usepackage{algorithm}
\usepackage{algorithmicx}
\usepackage{algpseudocode}
\usepackage{enumitem}
\usepackage{comment}
\usepackage{sidecap}

\usepackage[linewidth=1pt]{mdframed}
\usepackage[most]{tcolorbox}
\definecolor{block-gray}{gray}{0.85}
\newtcolorbox{shadequote}{colback=block-gray,grow to right by=-10mm,grow to left by=-10mm,
boxrule=0pt,boxsep=0pt,breakable}
\usepackage{cleveref}

\newtheorem{theorem}{Theorem}%
\newtheorem{lemma}[theorem]{Lemma}

\theoremstyle{remark}
\newtheorem{remark}{Remark}

\DeclareFontFamily{U} {MnSymbolA}{}

\DeclareFontShape{U}{MnSymbolA}{m}{n}{
  <-6> MnSymbolA5
  <6-7> MnSymbolA6
  <7-8> MnSymbolA7
  <8-9> MnSymbolA8
  <9-10> MnSymbolA9
  <10-12> MnSymbolA10
  <12-> MnSymbolA12}{}
\DeclareFontShape{U}{MnSymbolA}{b}{n}{
  <-6> MnSymbolA-Bold5
  <6-7> MnSymbolA-Bold6
  <7-8> MnSymbolA-Bold7
  <8-9> MnSymbolA-Bold8
  <9-10> MnSymbolA-Bold9
  <10-12> MnSymbolA-Bold10
  <12-> MnSymbolA-Bold12}{}

\DeclareSymbolFont{MnSyA} {U} {MnSymbolA}{m}{n}

\DeclareMathSymbol{\rcurvearrowright}{\mathbin}{MnSyA}{192}
\DeclareMathSymbol{\curvearrowright}{\mathbin}{MnSyA}{184}
\DeclareMathSymbol{\clockwise}{\mathbin}{MnSyA}{255}
\DeclareMathSymbol{\counterclockwise}{\mathbin}{MnSyA}{251}

\algnewcommand{\IIf}[1]{\State\algorithmicif\ #1\ \algorithmicthen}
\algnewcommand{\EndIIf}{\unskip\ \algorithmicend\ \algorithmicif}

\newcommand{\IndState}[1][1]{\State\hspace{#1\dimexpr\algorithmicindent}}

\algnewcommand\Assert[1]{\State \textbf{assert}(#1)}%

\newcommand{\ignore}[1]{}

\DeclareMathOperator*{\argmin}{arg\,min}

\newcommand{\CC}{\mathsf{CC}}

\newenvironment{theorem-repeat}[1]{\begin{trivlist}
\item[\hspace{\labelsep}{\bf\noindent Theorem \ref{#1} }]\em }%
{\end{trivlist}}

\newcommand{\mytime}{\mathsf{time}}

\newcommand{\ALGone}{Algorithm~1\xspace}
\newcommand{\ALGthree}{Algorithm~3\xspace}
\newcommand{\ALGfour}{Algorithm~4\xspace}

\newcommand{\CCinit}{\CC_{\mathsf{init}}}
\newcommand{\CCm}{\CC_{\mathsf{overhead}}}

\newcommand{\isTokenHolder}{\mathit{isTokenHolder}}

\newcommand{\Prev}{\mathit{prev}}
\newcommand{\Next}{\mathit{next}}
\newcommand{\Cycle}{\mathit{cycle}}

\newcommand{\data}{\textsc{data}\xspace}
\newcommand{\meta}{\textsc{end}\xspace}
\newcommand{\token}{\textsc{token}\xspace}
\newcommand{\request}{\textsc{request}\xspace}

\newcommand{\DATA}{\data}
\newcommand{\END}{\meta}
\newcommand{\TOKEN}{\token}
\newcommand{\REQ}{\request}

\setcopyright{none}
\copyrightyear{2022}
\acmYear{2022}
\acmDOI{}
\acmConference[]{}{}{}
\acmBooktitle{}
\acmPrice{}
\acmISBN{}

\begin{document}

\title{Distributed Computations in Fully-Defective Networks}
\thanks{e-mails: \texttt{\{ckeren,shirco,galy\}@cs.technion.ac.il, ran.gelles@biu.ac.il}}

\author{Keren Censor-Hillel}
\affiliation{%
  \institution{Technion}
  \city{Haifa}
  \country{Israel}
}
\email{ckeren@cs.technion.ac.il}

\author{Shir Cohen}
\affiliation{%
  \institution{Technion}
  \city{Haifa}
  \country{Israel}
}
\email{.}

\author{Ran Gelles}
\affiliation{%
  \institution{%
  Bar-Ilan University}
  \city{Ramat-Gan}
  \country{Israel}
}
\email{ran.gelles@biu.ac.il}

\author{Gal Sela}
\affiliation{%
  \institution{Technion}
  \city{Haifa}
  \country{Israel}
}
\email{.}

\begin{abstract}
We address \emph{fully-defective} asynchronous networks, in which all links are subject to an unlimited number of alteration errors, implying that \emph{all} messages in the network may be completely corrupted. Despite the possible intuition that such a setting is too harsh for any reliable communication, we show how to simulate any algorithm for a noiseless setting over any fully-defective setting, given that the network is 2-edge connected. We prove that if the network is not 2-edge connected, no non-trivial computation in the fully-defective setting is possible.

The key structural property of 2-edge-connected graphs that we leverage is the existence of an oriented (non-simple) cycle that goes through all nodes [Robbins, 1939]. The core of our technical contribution is presenting a construction of such a Robbins cycle in fully-defective networks, and showing how to communicate over it despite total message corruption. These are obtained in a \emph{content-oblivious} manner, since nodes must ignore the
content of received messages.
\end{abstract}

\maketitle

\hrule 
~

\linepenalty=8000

\section{Introduction}
\label{section:intro}
Faults are a main hurdle in a large variety of distributed systems.
Faults manifest themselves in several different manners, ranging from nodes that crash due to malfunctions to environmental disruptions that affect the communication channels connecting distant nodes. In the last few decades, research has focused on developing \emph{fault-tolerant} algorithms, as nodes crashes and channel noise are utterly inevitable.
See, e.g., recent books and surveys on fault-tolerant systems~\cite{Dubrova13,KK20} and algorithms~\cite{BDM93,Raynal18}, and references within.

In this work, we consider the case of channel noise within asynchronous distributed networks, where messages  communicated between nodes are subject to corruption. 
When dealing with channel noise, some restrictions must be imposed on its power. 
Clearly, if noise can affect channels arbitrarily without any restrictions, then it could, for instance, delete all the communication and prevent any non-trivial computation over the network. 
Previous work either limited the \emph{number} of channels that may suffer (arbitrary) noise~\cite{Dolev82,SW90,Pelc92,SAA95, HitronP21broadcast, HitronP21general} or the \emph{total amount} of corruptions (usually, alterations) the channels are allowed to make altogether~\cite{HS16,CGH19,GKR19,ADHS20}. 

Throughout this work, we consider noisy channels that may arbitrarily change the \emph{content} of transmitted messages, but can neither delete nor inject messages. This is known in the literature as alteration noise. %
Yet, we do not bound the amount of noise nor the number of noisy channels in any way. That is, we ask the following question:\\[-2.25ex]
\begin{shadequote}
Can one design fault-tolerant algorithms robust to an \emph{unlimited} amount of corruption on \emph{all} communication channels?
\end{shadequote}
On its surface, the above task seems doomed.
However, we answer the question in the affirmative for the large family of 2-edge-connected networks. We further show that if the network is not 2-edge-connected, the noise can destroy any non-trivial computation. 

Towards this goal, we develop \emph{content-oblivious} algorithms, that is, algorithms that do not rely on the \emph{content} of communicated messages~\cite{CGH19}. Instead, the actions of a node depend on the specific links and the order in which messages are received. 
In particular, we devise a method that compiles \emph{any} distributed algorithm into a content-oblivious version that computes the same task over 2-edge-connected graphs.

A folklore approach (see, e.g., in~\cite{JKL15,CGH19}) is to send a message along a certain path from~$u$ to~$v$ to signify a 0 bit, and to send a message along a different path to signify a 1 bit, where the existence of two different paths is promised by the 2-edge-connectivity property.
This approach conceals many challenges. First, the edges along these two paths are also edges in paths between other nodes in the network, and so the nodes must somehow be able to associate each such ``bit'' with its correct origin, in order to be able to decode each original piece of information and avoid mixing up bits of different ones. 
Second, in order to know where to forward the message to, the nodes need to extract the sender/receiver information from these ``bit'' messages, yet those might be fully corrupted.
Third, some guarantee needs to be obtained on the order in which different 0/1 ``bits'' arrive at their destination, in order for them to faithfully represent the encoded message, a caveat on which the asynchrony of the network imposes another obstacle. 

Before elaborating on how we overcome all these issues and stating our main results, let us explain our setting and noise model in more detail.
We abstract the network as a graph $G=(V,E)$ where every node~$v\in V$  is a computing device and every edge~$e\in E$ is a noisy bi-directional communication channel. 
Once $u$ sends a message~$m$ over some link~$(u,v)$, the channel guarantees that after some arbitrary yet finite time, $v$~receives some message $m'\in\{0,1\}^+$. Note that $m'$ may or may not equal~$m$. %
In other words, the noise over the channel can corrupt the content of any transmitted~$m$ into any~$m'$, but it cannot completely delete it, nor can it inject new messages. 
We say that~$G$ is a \emph{fully-defective} network if all its channels are noisy in the above manner.

\subsection{Our contribution and techniques}
\label{subsec:contributions}

\smallskip\textit{\textbf{Intuition: Content-oblivious encoding with parallel channels.}} 
Let us begin with a simple toy example that illustrates some of our techniques.
Suppose $u$ and $v$ are directly connected by \emph{two separate noisy channels}, which we name $\data$ and $\meta$. 
The basic idea is to communicate the information over the $\data$ channel by sending, ``bit-by-bit'', a unary encoding of the original message. 
In order to communicate the end of the unary encoding, a single message is sent on the $\meta$~channel.
Note, however, that timing is crucial: if the message sent on~$\meta$  is received before all the messages sent on~$\data$ reach their destination, the receiver decodes incorrect information. 
To avoid this confusion, the receiver sends one message over~$\meta$ as an acknowledgment for each received $\data$ message. The sender waits until all its $\data$ messages are acknowledged \emph{and only then} sends the termination message on~$\meta$.
Sending the terminating $\meta$ message has an additional effect: it switches the roles of the nodes. If $u$ is the sender, then after sending the $\meta$ message it takes the role of the receiver and vice versa. We call the sender at each point the \emph{token holder}.

\smallskip\textit{\textbf{Main result.}}
Since we do not wish to assume two separate channels between any two nodes, we ask whether they can be replaced with two separate \emph{paths} between any two nodes, i.e., can we constitute reliable communication between any two neighbors in 2-edge-connected graphs?

We answer this question affirmatively and show a method that takes any asynchronous message-passing distributed algorithm~$\pi$ for a noiseless network~$G$, and simulates it  over the fully-defective~$G$, given that $G$ is 2-edge connected. 
By \emph{simulating} we mean that every node has a black-box interface to~$\pi$ through which the node can deliver messages to~$\pi$ and (asynchronously) receive messages to be communicated to some neighbor. 
The simulation guarantees that,
at any given moment,
all the nodes behave similarly to some valid  execution of~$\pi$ over the noiseless~$G$.
\begin{theorem}[main, informal]
\label{thm:main-inf}
There exists a %
simulator for any asynchronous algorithm~$\pi$ such that
executing the simulator over a 2-edge connected fully-defective network~$G$ simulates an execution of~$\pi$ over the noiseless network~$G$.
\end{theorem}

Once we establish that such a simulator even exists, a natural question is, what is the best that could be aimed for in terms of its message overhead? 
To avoid excessive clutter in the presentation, 
we delay the complete statement of our main theorem that includes its overhead, to Theorem~\ref{thm:main} at the end of this section.

\smallskip\textit{\textbf{Warm-up: Resilient computations over a simple cycle.}}
To describe our approach for proving Theorem~\ref{thm:main-inf}, we begin with the much simpler case of \emph{cycle graphs}. 
In a cycle graph, any node $u$ 
is connected to only two neighbors. 
Our goal is to simulate $u$'s communication with its two neighbors over a fully-defective cycle.
In this special case, every two neighbors have exactly two separate paths between them: the direct link, and the rest of the cycle. %
The difference from the two-channel toy example illustrated above is that the paths of each two certain neighbors intersect the paths of other neighbors and we need to coordinate between the nodes so that each message reaches its correct destination and is interpreted correctly. 

We address this difficulty by guaranteeing that only a single node is the sender (token holder) at any given time. All the other nodes are passive and only forward messages along the cycle. In this way, the sender can communicate with its neighbor using both paths of the cycle---one of them, say, the clockwise path, replaces  the $\data$ channel, and the other one, the counterclockwise path, replaces the $\meta$ channel.

In fact, $u$ can use the same method to communicate with any other node on the cycle, since all the nodes see the same sequence of clockwise and counterclockwise messages. In a sense, $u$~broadcasts information over the cycle, and all the nodes learn this information. %

Next, we design a method to change the roles such that another node may become the token holder, i.e., the sender: After forwarding the message initiated by the sender in the $\meta$ path, the nodes enter a \emph{token delivery} phase.
During this phase, a counterclockwise~$\token$ message, initiated by the previous token holder, is forwarded along the cycle. When it reaches a node, if the node does not have a message to send, then it forwards the token along by propagating it counterclockwise. Otherwise, it becomes the new token holder and initiates \emph{a clockwise} message forwarded along the entire cycle which denotes the end of the token phase, so all the nodes go back to the stage of interpreting messages as $\data$ and $\meta$.

The above method has one significant drawback. %
If it takes some time for the nodes to produce a message to send, 
then the $\token$ message will keep circulating in the cycle, causing many superfluous transmissions. 
We circumvent this situation of wasteful transmissions by introducing a request mechanism: the token transfer is performed only if some node issues a request, which is done by sending a clockwise $\request$ message. The requesting node can be far away from the current token holder, thus, each node, upon receiving a $\request$ message, propagates it clockwise. We note that several nodes might issue a request at the same time at different locations on the cycle. 
Eventually, all nodes will have sent and received a $\request$ message, and after it reaches the current token holder, it issues the counterclockwise $\token$ message described above. 

This simulator for simple cycles, in which messages are interpreted as $\DATA$, $\END$, $\TOKEN$ or $\REQ$ based only on their direction and order of transmissions, is formally given and proved in Section~\ref{sec:cycle}.

\smallskip\textit{\textbf{Main result: Resilient computations over 2-edge-connected graphs.}}
To apply our approach for simple cycles to more complex graphs, we mimic it over a (not necessarily simple) cycle that goes through all the graph nodes. 
Such a cycle needs to be chosen carefully, because of the crucial role that the direction of messages plays in our approach.
Robbins's theorem~\cite{robbins39} states that any 2-edge-connected graph~$G$ is \emph{orientable}. That is, there exists a way to orient all edges in~$G$ so that the implied directed graph is strongly connected. 
This implies that there exists a cycle that goes through all nodes, possibly with multiple occurrences of some of the nodes, where all instances of any edge along the cycle bear the same orientation. We leverage the existence of such a Robbins cycle by mimicking our approach for the simple cycle over the Robbins cycle. To this end, we must first construct a Robbins cycle, as the nodes are unaware of the topology of the network. Then, we need to communicate over the Robbins cycle. Both steps are highly non-trivial and pose many challenges, as we now describe. During the first step---the construction of a Robbins cycle---we need the nodes to start communicating over partial pieces of the cycle (which are cycles by themselves), for which we need to already use the second step. For this reason, we describe the two steps in reverse order: we begin with describing the second step of communicating over a non-simple cycle, given that each node knows its previous and next neighbors along the cycle for each of its occurrences (\Cref{sec:GeneralCycle}). Then, we show the first step of how to construct the Robbins cycle and produce this information (\Cref{section:robbins}).

\smallskip\textit{\textbf{Second step: Communicating over a non-simple cycle.}}
The input of each node for this step, as will be guaranteed by our construction for the first step, is the previous and next nodes along the cycle for each of its occurrences. These inputs are consistent with some Robbins cycle, so that, in particular, each edge in the cycle has a unique orientation, and a single orientation of the edges is considered by all nodes as the \emph{clockwise} direction.

Mimicking our approach for a simple cycle over a Robbins cycle brings along several challenges. Consider, for instance, the network~$G$ and its induced Robbins cycle depicted in Figure~\ref{fig:RobbinsCycle}.  
Suppose a clockwise message is received at node~$d$ along the edge $(c,d)$. Should this message be propagated over the edge~$(d,e)$ or over the edge~$(d,a)$, or maybe over both? 
Note that both these options are in the clockwise direction, however, they belong to different segments of the cycle. Further, note that some messages are initiated in an asynchronous manner, e.g., the $\request$ message. Thus, when the node~$d$ receives a $\request$ message from node~$c$, it is possible that the request originated at node~$a$ and should be propagated to node~$e$ or it originated at node~$e$ and should be propagated to node~$a$.

\begin{SCfigure}[0.51][ht]
    \centering
    %

\quad\quad
\begin{tikzpicture}

\def \n {4}
\def \radius {1.3cm}
\def \margin {10} 

\foreach \s in {1,...,\n}
{
  \node[draw, circle] (v\s) at ({360/\n * (\s - 1)}:\radius) {\phantom{$v_\s$}};
}
\node[draw,circle] (v5)  {\phantom {$v_5$}};

\draw[->,>=latex] (v1) -- (v2);
\draw[->,>=latex] (v3) -- (v2);
\draw[->,>=latex] (v4) -- (v1);
\draw[->,>=latex] (v4) -- (v3);

\node at (v3) {$a$};
\node at (v2) {$b$};
\node at (v5) {$c$};
\node at (v4) {$d$};
\node at (v1) {$e$};

\draw[->,>=latex] (v2) -- (v5);
\draw[->,>=latex] (v5) -- (v4);

\end{tikzpicture}
~~~~~~~\phantom{XXXXXX}~~~~~~~
\begin{tikzpicture}

\def \n {8}
\def \radius {1.35cm}
\def \margin {13} 
\def \innerfrac {0.65}

\foreach \s in {1,...,\n}
{
  \node[draw, circle] (v\s) at ({360/\n * (\s - 1)}:\radius) {\phantom{$v_{\s}$}};
  \draw[<-, >=latex] ({360/\n * (\s - 1)+\margin}:\radius) 
    arc ({360/\n * (\s - 1)+\margin}:{360/\n * (\s)-\margin}:\radius);
}

\node at (v1) {$c$};
\node at (v2) {$b$};
\node at (v3) {$a$};
\node at (v4) {$d$};
\node at (v5) {$c$};
\node at (v6) {$b$};
\node at (v7) {$e$};
\node at (v8) {$d$};

\end{tikzpicture}
    \caption{(a) A 2-edge-connected graph~$G$ with a Robbins orientation and (b) the resulting Robbins cycle with multiple occurrences per node. The arrows denote the clockwise direction of the cycle. 
    \\ \\ }
    \label{fig:RobbinsCycle}
\end{SCfigure}
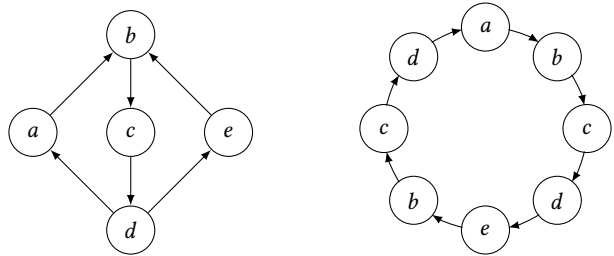

We cope with these issues using two separate mechanisms. The first mechanism makes sure that the $\token$, $\data$, and $\meta$ messages are propagated correctly along the Robbins cycle. This mechanism consists of two main ingredients.
First, we guarantee that these three message types are forwarded in a sequential manner, in the sense that the token holder issues the next message among them only after receiving the previous one from the other direction of the cycle.
Second, we assure that at any given moment, each node~$u$  knows ``where the token is'', that is, on which segment of the cycle (i.e., between which two occurrences of~$u$) the token resides. 
Since the token holder is the only node to initiate the above three message types, knowing the relative position of the token holder resolves the above and allows each node to track each message along the Robbins cycle. Indeed, each such message must first arrive from the segment in which the token holder resides and then be propagated by~$u$ to the next segment in the respective direction of the cycle. %
We prove that since at any given moment only one message travels through the cycle, there can be no confusion at~$u$ regarding what the message type is, which one of $u$'s occurrences has received a message and where a message should be forwarded to.

The second mechanism we employ is for $\REQ$ messages. 
These have no pre-specified origin, and they can be initiated by any node and even by multiple nodes at the same time. The mechanism for these messages is as follows. Whenever a node receives a $\request$ message or when a node wishes to initiate one, it sends a clockwise message to \emph{all} of its clockwise neighbors along the Robbins cycle at the same time. Then, the node waits to receive a $\request$ message from each of its counterclockwise neighbors \emph{and only then} it continues with executing the cycle algorithm described above. We prove that this guarantees that {all nodes} send and receive a request message \emph{regardless} of their position(s) on the Robbins cycle.

\smallskip\textit{\textbf{First step: A content-oblivious construction of a Robbins cycle.}}
Our Robbins cycle construction follows an \emph{ear-decomposition} technique 
by Whitney~\cite{Whitney32}, 
claiming that 
any 2-edge-connected graph $G$ can be decomposed into edge-disjoint parts, 
\(
G=C_0 \cup E_0 \cup E_1 \cup \cdots \cup E_k,
\)
where $C_0$ is a simple cycle, and for any $0 \leq i \leq k$, $E_i$ is an \emph{ear}---a simple path or cycle whose endpoints belong to $C_0 \cup E_0 \cup \cdots \cup E_{i-1}$. 
Following Whitney's work, we iteratively decompose a 2-edge-connected graph~$G$ into some $C_0,E_0,\ldots,E_k$, part by part, and combine them into a Robbins cycle.  
The main obstacle we face is that our construction must be content-oblivious and 
cannot rely on the content of messages sent by the nodes. 

The first stage of our construction is performing a DFS-like search starting from a specified \emph{root} node.
The DFS search progresses by sending a message (\emph{a DFS-token}) sequentially, i.e., each node propagates this message to one of its unexplored adjacent edges.
This DFS-token message propagates through the network until it reaches the root node again. 
At this stage, the path the DFS-token has taken defines a cycle~$C_0$.  

The key challenge in this stage is that the DFS-token might reach some node~$u$ twice before reaching the root. This might cause the DFS to ``get stuck'', e.g., if $deg(u)=3$. We overcome this pitfall by insisting on  $C_0$ being a simple cycle that starts and closes at the root. If some $u\ne root$ receives the DFS-token for the second time, it sends that message back on the same edge on which it was received. This has the effect of ``backtracking'' that edge so it is excluded from the constructed cycle. Nodes that backtrack all their adjacent edges go back to their initial state and are added to the Robbins cycle at a later step.

Once $C_0$ is established, the nodes on it switch to the second stage, in which they use our resilient communication approach of the above second step, in order to coordinate exploring further ears. One node on~$C_0$ that has adjacent edges that do not belong to~$C_0$ gets selected to initiate another DFS-like search, which again propagates in~$G$ until reaching a node on~$C_0$, possibly different from the initiator. The path the DFS-token takes defines the ear~$E_0$. 
Then, the nodes on~$C_0$ and~$E_0$ jointly coordinate to form a new non-simple cycle~$C_1$ that includes all the edges in~$C_0$ and~$E_0$. %
The nodes on~$C_1$ switch to communicate over this cycle using the above resilient communication of the second step. 
The nodes iterate this process, until a Robbins cycle is formed. A crucial aspect of these iterations of adding ears is that much coordination is required among the nodes for switching in a timely manner from communication on $C_i$ to communication on $C_{i+1}$. The technical specification of this mechanism is given in \Cref{section:robbins}. 

We emphasize that the nodes do not know $|V|$, and hence they do not know when a Robbins cycle is already formed, i.e., when each node already appears on the current~$C_i$ at least once. 
Instead, they keep adding edges to the constructed cycle, until no node has an adjacent edge that is not in~$C_i$, which is a state they can detect. At that point, the construction ends.

\smallskip\textit{\textbf{Putting it all together.}}
With the above two steps, our result can now be formally stated. Given any 2-edge-connected fully-defective network~$G$ and an asynchronous algorithm~$\pi$ designed to work on the noiseless~$G$, 
we show how to compute $\pi$ over the fully defective~$G$ by first constructing a Robbins cycle~$C$ on~$G$ using a resilient content-oblivious algorithm, and then simulating $\pi$ over the Robbins cycle~$C$ in a resilient content-oblivious manner. 
\begin{theorem}[main]
\label{thm:main}
There exists a %
simulator for any asynchronous algorithm~$\pi$, such that
executing the simulator over a 2-edge connected fully-defective network~$G$ simulates an execution of~$\pi$ over the noiseless network~$G$.

The simulator has a pre-processing phase that construct a Robbins cycle~$C$ (which depends only on~$G$) and an online phase that simulates the communication of~$\pi$ over~$C$. 
The pre-processing step communicates $\CCinit=|C|^{O(1)}$ bits. In the online phase, any message~$m$ communicated by $\pi$ is simulated by communicating $\CCm(m) = O(|C|\cdot |m|+|C|\log |V|)$ bits. 
\end{theorem}
We note that, in the worst case, $|C|=O(|V|^3)$; see Section~\ref{sec:robbins-cycle-length} for a detailed  discussion.
We do not strive to optimize the polynomial overhead of our schemes, as their mere existence is the focus of this paper. Nevertheless, unary encoding as explained above imposes an exponential overhead in the length of the message. In 
\Cref{sec:binary}
we offer a binary encoding method that reduces the communication complexity to the polynomial terms stated above.

\smallskip\textit{\textbf{Impossibility result.}}
We complement the above result and show that if $G$ is not 2-edge-connected, then there is no way to conduct non-trivial computations over a fully-defective~$G$. 
To this end, we prove the following impossibility for \emph{two-party computation} over a fully-defective channel. The impossibility for a non 2-edge-connected $G$ follows since it contains a bridge, and we can reduce the two sides of the bridge to the two-party case.

\begin{theorem}
\label{thm:main-inf-cycle} 
Fix a non-constant function~$f(x,y)$.
No two-party deterministic algorithm that gives output or terminates
can compute~$f$ over a fully-defective channel.
\end{theorem}
The theorem requires the nodes to either terminate or irrevocably give an output. Note that the above theorem differs from the famous \emph{two generals coordinated-attack} impossibility~\cite{Gray78}, since our noise model does not allow deleting messages. 
See \Cref{sec:impossibility} for complete details.

\subsection{Related work}
\label{sec:relatedWork}

There are two common ways to deal with channel corruptions. One is by adding redundancy, i.e., coding the information, an approach that is known in the literature as \emph{Interactive Coding}.
The other is by diverting the communication so it would not pass through corrupted edges, which are known as \emph{Byzantine edges}. 

We review some related work in these areas, but we stress that neither approach can be used in fully-defective networks: Interactive coding must assume some bound on the errors, either per channel or globally, while solutions for networks with Byzantine edges must assume a bound on the number of noisy channels.

Interactive coding was initiated by the seminal work of Schulman~\cite{schulman92,schulman93,RS94}, see \cite{gelles17} for a recent survey on this field. 
In this setting, communication channels either suffer  from stochastic noise~\cite{RS94,GMS14,BEGH17,GK19,ABEGH19} or from some bounded amount of adversarial noise. 
E.g., if limiting the overhead of the coding scheme to be linear, \cite{GMS14,HS16,JKL15,LV15,GKR19} develop schemes resilient to up to a fraction $O(1/|E|)$ of the total communication. Without any restriction on the overhead, schemes can cope with noise up to a fraction~$O(1/|V|)$ of the total communication, and such a fraction is shown to be maximal~\cite{JKL15}---otherwise, the adversarial noise could completely corrupt all the outgoing communication of the node that communicates the least.
The above works assume synchronous networks. 
Censor-Hillel, Gelles, and Haeupler~\cite{CGH19} developed the first coding scheme for \emph{asynchronous} networks that suffer from up to a fraction~$O(1/|V|)$ of adversarial noise. 
Communication with an unbounded (yet, \emph{finite}) amount of noise was examined in \cite{DMSY15,ADHS18,GI19} for the two-party case and in
\cite{ADHS20} for the multiparty case.
In a work by Efrmenko, Haramaty, and Kalai~\cite{EHK20}, the noise model is similar to the one we consider here in the sense that it can corrupt \emph{the content} of messages but not their existence. However, the amount of bit-corruptions in~\cite{EHK20} (measured as the edit distance between sent and received messages) is bounded to a constant fraction out of the entire communication. Furthermore, their work considers only two parties.

Networks with Byzantine edges do not restrict the amount of noise per link, and even allow insertion/deletion errors, but 
allow only a bounded number of links to be noisy. 
In asynchronous settings, %
Fisher, Lynch, and Paterson~\cite{FLP85} exclude the existence of consensus algorithms when a single node may crash, or equivalently, when all the links connected to some single node may crash. 
In synchronous networks, certain tasks are also impossible with arbitrary link failures~\cite{Gray78,SBK09}.
On the other hand,
Santoro and Widmayer~\cite{SW90} considered distributed function evaluation when (a large number of) links suffer either corruptions, insertions, deletions, or their combination. In a sense, the synchrony guarantee allows simpler solutions, e.g., encoding information via the time in which messages are sent.
Pelc~\cite{Pelc92} shows that if the number of Byzantine links is bounded by~$f$, robust communication is achievable only over graphs whose edge-connectivity is more than~$2f$. 
This is also implied by the work of Dolev~\cite{Dolev82}. Additional works~\cite{PT86,gong1998byzantine,SCY98,Dasgupta98,biely03} consider the case of mixed node and link failures.

Recent work by Hitron and Parter~\cite{HitronP21broadcast,HitronP21general} gives a compiler that turns any algorithm in the noise-free setting into an algorithm that works correctly even if the adversary controls $f$~edges in a $(2f+1)$-edge-connected network. The above is for the synchronous Congest setting. Their approach is to construct a family of low-congestion cycle-covers (see also~\cite{ParterY19,PY19}), which are structures in which for every edge~$(u,v)$, there are at least $f+1$ cycles that contain no adversarial edges. We stress that low-congestion cycle-covers do not seem to be helpful for our setting: Even if we were promised only \emph{two} cycles that share a \emph{single} edge, it is not clear how to communicate over them in a way that distinguishes one from the other.

\section{Preliminaries}\label{section:preliminaries}

\smallskip\textit{\textbf{Notations.}}
We use $a\|b$ or $a\cdot b$ for the concatenation of $a$ and~$b$. For a positive integer $k\in \mathbb{N}$ and a string~$b$, we let $b^k=b\cdot b \cdot \cdots \cdot b$ denote $b$ concatenated to itself for $k$ times; $b^0=\epsilon$ is the empty string.
For a string $b$ and an integer $0\le i\le |b|-1$, we let $b_i$ denote the $i$-th bit of~$b$, i.e., $b=b_0b_1\cdots b_{|b|-1}$.

\smallskip\textit{\textbf{Networks and protocols.}}
A protocol~$\pi$ over an undirected network~$G=(V,E)$ with $n=|V|$ nodes is an asynchronous event-driven distributed algorithm, in which nodes conduct some computation by sending messages to their neighbors in~$G$ (for simplicity, we assume only deterministic algorithms in this paper). Upon the reception of a message,  $\pi$ instructs the recipient node what message(s) to send next, as a function of the node's input and all the messages it has received so far. 
Specifically, each node~$v$ begins with a private input~$x_v$ (which may be empty), and knowledge of the IDs of its neighboring nodes, $N(v) = \{u \mid (u,v)\in E\}$ 
(we can remove this assumption, see Remark~\ref{rem:KT0}). 
According to the input to~$v$, $\pi$ generates messages to send to zero or more of $v$'s neighbors (possibly different messages to different neighbors). Afterwards, the protocol behaves in an event-driven manner, i.e., nodes act only upon receiving messages: 
whenever a node~$v$ receives a message, it performs some computation and produces messages designated to zero or more of its neighbors. We impose no assumption on the computation time of~$\pi$ except that it is finite.
We additionally assume a preselection of one designated node (which will function as a root node in our Robbins cycle construction), and assume that every node knows whether it is the designated node.

Communicating a message over some link of $G$ takes arbitrary positive finite time. Channels are \emph{not} assumed to be FIFO.
Incoming messages are kept in an incoming buffer until processed by the node.

The \emph{protocol's transcript} $\tau$ of a given execution, is the sequence of messages sent and received during the execution. Each item in~$\tau$ indicates the message sent or received, the sending or receiving node and the link on which the message was communicated. 
Events that happen in different nodes at the same time   
are assumed to be ordered in some arbitrary order. 
The \emph{local transcript} $\tau_v$ of a node~$v$,
is the ordered sequence of messages sent and received by~$v$. Note that $\tau_v$ can be derived from~$\tau$ as the sub-sequence in which $v$ is the sending or receiving node.

We say that $\pi$ \emph{gives an output} if every node eventually writes an output to its write-only output register. This action is irrevocable. If needed, the node may remain active and send and receive messages after giving an output; that is, we do not require termination, but our result also applies to protocols that terminate.
We say that the protocol has reached \emph{quiescence} at some time, if no message is still in transit and from that time on, no new messages are sent over the network.

\smallskip\textit{\textbf{Fully-defective networks and noise-resilient simulations.}}
We work in networks with noisy channels exposed to alteration noise, which can corrupt the \emph{content} of any message communicated over any channel. 
That is, once a message $m\in\{0,1\}^+$ is sent over some link, the received message may be any $ m'\in\{0,1\}^+$. 
However, the noise \emph{cannot} completely delete a message 
nor can it inject a message on a link in which no $m$ was sent. We stress that, except for inserting and deleting messages, the noise has no restrictions at all. In particular, it can apply to all channels and corrupt all messages in a given execution. 
We call networks that suffer noise as specified above  \emph{fully-defective} networks.
Equivalently, one can think about such a network as one in which nodes communicate only by means of sending pulses to their neighbors, which could be the case, for instance, when the nodes have very basic communication hardware.  

A \emph{noise-resilient simulator} 
designed for a noiseless network $G=(V,E)$
is a protocol~$\hat \pi$ which is given as an input an asynchronous black-box interface to some~$\pi$.
When $\hat \pi$ is executed on a fully-defective network~$G$, it produces for each node~$v\in V$ a 
string~$\hat\tau_v$, such that there exists some execution of~$\pi$ over the noiseless network~$G$ that generates a transcript~$\tau$, for which $\tau_v = \hat\tau_v$ for each node~$v$.
We allow a simulator to perform some pre-processing before simulating~$\pi$. We define $\CCinit$ to be the communication complexity in bits of the simulator during the pre-processing, and $\CCm(m)$ to be the communication complexity for simulating the delivery of a message~$m$. Note that $\CC$ accounts only for the length of \emph{sent} messages, even if later their content is corrupted by the noise.

\smallskip\textit{\textbf{Distributed representation of cycles.}}
A (directed) cycle can be represented in a distributed network in two manners: \emph{locally} and \emph{globally}.
A local representation of some cycle~$C$ means that every node on~$C$ knows its two neighbors on the cycle along with their respective direction, clockwise or counterclockwise, usually held in the local variables $\Next$ and $\Prev$, respectively. 
In case $C$~is not a simple cycle, then every node knows its clockwise and counterclockwise neighbors for each of its occurrences on~$C$.
This information is consistent across all nodes in the sense that an outside observer who follows the neighbors and directions of each node would see a consistent directed cycle.

A global representation of a directed cycle means that every node $v\in C$ holds the string $C=(v_1, v_2, \ldots)$ of the IDs of the nodes on~$C$ in their clockwise order.

\section{Simulating computations over a fully-defective simple cycle}
\label{sec:cycle}

As discussed in Section~\ref{section:intro}, we can establish a resilient connection between two nodes connected by two separate links, sending \emph{content-less} messages between them, which we will call \emph{pulses} throughout this paper. 
Our goal is to implement this idea for any two nodes in a 2-edge-connected graph, since in such a graph any two nodes are connected by two separate \emph{paths}. 
As a stepping stone, in this section we consider the special case of \emph{simple cycles}.

\begin{theorem}[A simulator for a simple cycle]
\label{thm:simpleCycle}

There exists a noise-resilient simulator for any asynchronous protocol~$\pi$ and any fully-defective simple cycle~$G$ in which each node knows its clockwise and counterclockwise neighbors.
The simulator features  $\CCinit=0$ and 
$\CCm(m) = O(|V|\cdot |m|+|V|\log |V|)$ pulses.
\end{theorem}

Let $G$
be a simple cycle on $V=\{v_i\}_{0\le i \le n-1}$ with $E=\{ (v_i,v_{i+1})\}_{0 \le i \le n-1}$, where indices are taken mod~$n$. The main idea is to imitate the two-channel idea described in \Cref{section:intro} above over the cycle. 
That is, suppose $v_i$ wishes to send a message to its neighbor~$v_{i+1}$. We can think of the link $(v_i,v_{i+1})$ as the $\data$ channel, and on the \emph{path} $v_i,v_{i-1},v_{i-2},\ldots, v_{i+1}$ as the $\meta$ channel. For this to work, all the nodes beside $v_i$ and $v_{i+1}$ need to simply forward each pulse they receive along the same direction. 
However, the above description supports only a single fixed sender and a single fixed receiver. Thus, we need a method that allows different nodes to become the sender. For this we use a token mechanism, where only a single node holds the token at any given time.

Our simulator can be split into two separate phases per message transfer: the first one is the \emph{token phase} which handles transferring the token between the nodes, and the second one is the \emph{data phase} that handles communication between the current token-holder and the rest of the nodes. 

The token phase works as follows. At the starting point, there exists only a single token holder.
During the token phase, pulses carry one out of two possible meanings: either they are a $\REQ$ pulse or a $\TOKEN$ pulse. 
The meaning of a pulse is dictated by the direction in which the pulse progresses along the cycle: $\REQ$ is a clockwise pulse while $\TOKEN$ is a  counterclockwise pulse. 
A node that wishes to obtain the token issues a $\REQ$ pulse. 
Every node that receives such a $\REQ$ pulse, propagates it in the same direction, unless it has already sent a $\REQ$ pulse previously in this phase, so eventually every node sends and receives \emph{a single} $\REQ$ pulse.

Upon receiving a $\request$ pulse, the current (single) token holder releases the token by sending a counterclockwise $\TOKEN$ pulse. 
This pulse propagates along the cycle until it reaches one of the nodes which requested the token. A node that receives the $\TOKEN$ pulse and wishes to become a token holder does not propagate the $\TOKEN$ pulse but instead sets itself as the  new token holder.
Then, the new token holder switches to its data phase and begins sending clockwise pulses, which are interpreted as $\DATA$ pulses. 
The first of these pulses propagates throughout the entire cycle and informs all the other nodes that the token phase has completed. This first pulse cannot be confused with a $\REQ$ pulse since we guarantee that every node sends and receives exactly a single $\REQ$ pulse in each token phase. In other words, the second clockwise pulse received during a token phase must be a $\DATA$ pulse, which triggers its recipient to switch to its data phase.

In the data phase, the token holder delivers its message via a unary coding. That is, it sends a number of clockwise $\DATA$ pulses that equals the length of the unary encoding of the information. Each node other than the token holder forwards each received $\DATA$ pulse clockwise, so these pulses propagate along the cycle until they reach the token holder back from the other side of the cycle. Then, the token holder sends a single counterclockwise $\END$ pulse that signals the end of the message and the end of the data phase. Note that once the token holder receives the $\END$ pulse from the other direction, all nodes know that the data phase is over, and are back in the token phase. Note also that due to the asynchrony, nodes that already moved to the next token phase might send a $\REQ$ pulse before the $\END$ pulse arrives at the token holder. Our design promises that these $\REQ$ pulses are not confused with pulses of the current data phase: $\END$ pulses are sent in the other direction, and as for $\DATA$ pulses---the token holder does not proceed to sending a $\REQ$ pulse before it receives the $\END$ pulse of the data phase, so $\REQ$ pulses of the new token phase can only reach nodes that have already received the $\END$ pulse for this phase and therefore do not interpret them as additional $\DATA$ pulses.

A \emph{phase} is a local concept, in the sense that each node runs a specific data or token phase in any given time, and different nodes might be in different phases in a certain time. We denote each token phase and its subsequent data phase an {\em epoch}. An epoch is a local concept too, viewed by each node according to the phase it is currently running. Different nodes might be in different epochs in a certain time: some nodes might already send a $\REQ$ pulse in the new epoch while others have still not received an $\END$ pulse for the previous epoch.

\subsection{Formal description}
We now formally describe our simulator over fully-defective simple cycles, where each node is given the identities of its clockwise and counterclockwise neighbors.
Our simulator receives as an input an asynchronous protocol~$\pi$ for noiseless communication channels. Messages to be sent are generated by~$\pi$, and any message received by a node in our simulator is delivered and processed by~$\pi$. Our simulator thus treats~$\pi$ as an asynchronous black box that interfaces with the simulator by sending and receiving messages, internally at each node. We stress that $\pi$'s actions take finite arbitrary time unknown to and independent of the simulator algorithm. 

Our simulator appears in \Cref{alg:cycle compiler token phase,alg:cycle compiler data phase}. 
All nodes begin executing the token phase (\Cref{alg:cycle compiler token phase}).
Each node $u$ has an internal $\isTokenHolder_u$ variable that indicates whether it is the token holder.
Moreover, each node $u$ keeps a %
queue~$Q_u$ of messages generated by~$\pi$, which should be broadcast over the cycle.
Messages in~$Q_u$ are of the form $(m, u, v)$, where $m$ is a message that $\pi$~ instructs~$u$ to send to~$v$.
At the onset, $\isTokenHolder_u$ is $\mathsf{True}$ for a single node, and each $Q_u$ is empty. 
When $\pi$ gives an output, the respective node gives the same output in the simulator but keeps executing the communication algorithm over the cycle. If in a certain time all the queues $\{Q_u\}$ are empty and remain empty, then the simulator stops sending messages and reaches quiescence.

The simulator is content-oblivious, and as such it communicates by sending pulses (content-less messages).
Note that in our algorithms we write next to each pulse its meaning ($\DATA,\END,\REQ,\TOKEN$), however, this is only for the analysis; the nodes assign this meaning according to their current state and the clockwise/counterclockwise direction of the pulse, and not by the pulse content, which is ignored.

\renewcommand\thealgorithm{\arabic{algorithm}(a)}   %

\begin{algorithm}[htp]
\caption{A simulator for simple cycles: token phase (node $u$)}\label{alg:cycle compiler token phase}
\begin{algorithmic}[1]
\Statex \textbf{Init:} 
A single node has $\isTokenHolder = \mathsf{True}$. Node~$u$ holds a (possibly empty) input $x_u$ for~$\pi$. %
\Statex \textbf{Handling messages sent by $\pi$:}
During the execution of the algorithm, node $u$ enqueues to $Q_u$ any new message $\pi$ asks $u$ to send, in the form (message, source, destination). The actions of~$\pi$ occur in parallel to the execution of this algorithm.

\Statex

\State \textbf{wait until} $Q_u$ is  not empty or a clockwise $\REQ$ pulse is received
\label{line:cycle token phase - wait for REQ or new Qu message}

\State send a clockwise $\REQ$ pulse\label{line:cycle token phase - send REQ}
    
\IIf {%
no clockwise $\REQ$ pulse was received}
    wait until receiving a clockwise $\REQ$ pulse 
    \label{line:cycle token phase - wait for REQ after seding one}
\EndIIf\label{line:token-phase-end-REQ-clause}
\If {$\isTokenHolder_u$}\label{line:cycle token phase - release token}
    \State $\isTokenHolder_u \gets \mathsf{False}$ \label{line:cycle-tokenFalse} 
    \State send a counterclockwise $\TOKEN$ pulse \label{line:cycle token phase - send TOKEN}
\EndIf\label{line:cycle token phase - release token-end if}

\State \textbf{wait until} receiving a pulse \label{token-phase-wait-after-REQ}
\If {the pulse is a counterclockwise $\TOKEN$ pulse} \Comment{Else, the pulse is a clockwise $\DATA$ pulse}
    \label{token-phase-wait-after-REQ-2}
        \If {$Q_u$ is not empty}
            \State $\isTokenHolder_u \gets \mathsf{True}$ \label{token-phase-set-tokenholder}
        \Else
            \State forward the counterclockwise $\TOKEN$ pulse 
            \State \textbf{go to} \Cref{token-phase-wait-after-REQ}
        \EndIf

    \EndIf\label{line:c-end}
\State continue with \Cref{alg:cycle compiler data phase} 
\label{line:cycle token phase:end}
\end{algorithmic}
\end{algorithm}

\renewcommand\thealgorithm{\the\numexpr\value{algorithm}-1}
\renewcommand\thealgorithm{\the\numexpr\value{algorithm}-1(b)}   %

\begin{algorithm}[htp]
\caption{A simulator for simple cycles: data phase (node $u$)}\label{alg:cycle compiler data phase}
\begin{algorithmic}[1]

\setcounter{ALG@line}{\getrefnumber{line:cycle token phase:end}}

\If{$\isTokenHolder_u$} %
    \State dequeue a message from $Q_u$, denote it by $(m,u,v)$ and let $1^d$ be its unary encoding \label{line:cycle data phase - dequeue}
    \State send $d$ clockwise $\DATA$ pulses \label{line:cycle data phase - initiate data}
	\State \textbf{wait until} receiving $d$  clockwise $\DATA$ pulses 
	\State send a counterclockwise $\END$ pulse\label{line:cycle data phase - initiate end}
	\State \textbf{wait until} a counterclockwise $\END$ pulse is received\label{line:cycle data phase - wait for end}
\Else %
	\State forward any received clockwise $\DATA$ pulse \textbf{until} receiving a counterclockwise $\END$ pulse \Statex \Comment{Including the $\DATA$ pulse received during the preceding token phase} \label{line:cycle data phase - forward DATA}
	\State let $count$ be the number of received clockwise $\DATA$ pulses
	\label{line:cycle data phase - record message}
	\State decode $1^{count}$ as the unary encoding of $(m',u',v')$
	\label{line:cycle data phase - decode}
	\IIf{$u = v'$} 
	    deliver $m'$ to $\pi$ (as if received from $u'$) \label{line:cycle data phase - deliver}
    \EndIIf\label{line:cycle data phase - complete process message}
	\State forward the counterclockwise $\END$ pulse
	\label{line:cycle data phase - forward END}
\EndIf
\State continue with Algorithm~\ref{alg:cycle compiler token phase} 
\label{line:cycle data phase:end epoch}
\end{algorithmic}
\end{algorithm}

\renewcommand\thealgorithm{\the\numexpr\value{algorithm}-1}

\subsection{Analysis}
\label{section:analysis-cycle}

Let us set some notation for the analysis of \ALGone.
Let `1' indicate a pulse sent clockwise, and let `0' indicate a pulse sent counterclockwise. Recall that an \emph{epoch} is the execution of consecutive token and data phases. We say that a node has completed its $k$-th epoch once it has executed \Cref{line:cycle data phase:end epoch} 
for the $k$-th time.
Let $T_k$ be the $k$-th node to have set its $\isTokenHolder_u$ to $\mathsf{True}$ in \Cref{token-phase-set-tokenholder}, 
whereas $T_0$ is the node whose $\isTokenHolder_u$ variable is initialized to $\mathsf{True}$.
(We will show that $T_k$ sends, in its $k$-th epoch, the $k$-th simulated message in the system.)
Let $s_{k}$, for $k\geq1$, be the time in which $T_k$ sets its $\isTokenHolder \gets \mathsf{True}$ ($s_k=\infty$ if $T_k$ is undefined). 
Finally, let $t_k$ be the time in which $T_k$ completes its $k$-th epoch ($t_k=\infty$ if $T_k$ is undefined or never ends the $k$-th epoch). We let $s_0=t_0=0$.  
Let $[u \curvearrowright v]$ denote the clockwise path from~$u$ to~$v$ along the cycle including both ends, and similarly let $[u \rcurvearrowright v]$  denote the counterclockwise path from $u$ to $v$. 
In the special case of identical endpoints, $[u \curvearrowright u]$ denotes the path through the whole cycle.
To exclude an endpoint, we use a round bracket in place of a square bracket, e.g. $[u \curvearrowright v)$ denotes the clockwise path excluding $v$; the path can be empty, i.e., $(u \curvearrowright v)$ for $u,v$ neighbors.

Our analysis is based on the following technical lemma, which provides us with three important properties satisfied by  \ALGone in every epoch: 
(1) \textbf{progress,} which says that as long as there is a message to send, the next epoch will eventually start and complete; 
(2) \textbf{single token holder,} which says that at most a single node holds the token at any moment (there is no such node during the time in which the token is being passed), and $T_k$ is the only one to hold it during the data phase of the $k$-th epoch; and 
(3) \textbf{global consistency,} which says that in any given epoch~$k$, exactly one message is being communicated---sent by~$T_k$ and received by all other nodes, and the pattern of pulses every node sends has a distinct structure.
We now formalize these ideas as follows.
\begin{lemma}
\label{lemma:cycle}
Consider an execution of \ALGone
and consider any ${k \ge 1}$, 
for which ${t_{k-1}<\infty}$.
If from time $t_{k-1}$ and forward, all queues $\{Q_v\}_{v\in V}$ are always empty, then $t_{k}=\infty$.
Otherwise,
the following hold:

\begin{enumerate}[label = (\arabic*)]

\item \label{Progress} \textbf{Progress:} 
All nodes eventually complete their $k$-th epoch. 
In particular, $t_{k}<\infty$, and 
at time~$t_k$, all nodes have already processed the $\END$~pulse (of epoch~$k$) but have not yet passed \Cref{token-phase-wait-after-REQ} in epoch~$k+1$
(they are either waiting in Lines \ref{line:cycle token phase - wait for REQ or new Qu message} or \ref{line:cycle token phase - wait for REQ after seding one} for a $\REQ$ pulse, or waiting in \Cref{token-phase-wait-after-REQ} for either a $\DATA$ or a $\TOKEN$ pulse).

\item \label{Single token holder}\textbf{Single token holder:}
It holds that $t_{k-1} < s_k < t_k$.
At each moment in~$(t_{k-1},t_k)$, there is at most a single node for which  $\isTokenHolder=\mathsf{True}$.
More specifically, within this time frame, the token is passed as follows: the node $T_{k-1}$ releases the token at some time in~$[t_{k-1},s_{k})$ and the node~$T_k$ is the next node that gains the token at time $s_k$. The node $T_k$ (solely) holds the token in~$[s_k,t_k]$.

\item \label{Global consistency} \textbf{Global consistency:} 
There exist integers $d_1,\dots, d_k > 0$ and for any $u\in V$ there are $b^u_1,\dots, b^u_k \in \{0,1\}$, such that when the node~$u$ completes its $k$-th epoch, its sent transcript (the overall pulses sent so far by~$u$) is $P_{u,k} \triangleq 10^{b^u_1}1^{d_1}0\cdot10^{b^u_2}1^{d_2}0\dots10^{b^u_k}1^{d_k}0$.

In addition, the message each node decodes and processes (\Crefrange{line:cycle data phase - record message}{line:cycle data phase - complete process message})
in its $k$-epoch is the unary decoding of $1^{d_k}$, which is the message sent by $T_k$ (\Cref{line:cycle data phase - dequeue}) in its $k$-th epoch.
\end{enumerate}
\end{lemma}

\begin{proof}%
We prove the statement by induction on the epoch number~$k$.
We start with proving the base case, $k=1$. The proof for the general case is very similar. 
The analysis follows the progress of the protocol and shows that each pulse sent with a certain meaning (i.e., $\DATA$, $\END$, $\TOKEN$, $\REQ$) is correctly interpreted by its recipient.

\textbf{Base Case, $k=1$.}
Note that $t_0=0$, hence, $t_{k-1}<\infty$.
All the nodes begin by executing \Cref{alg:cycle compiler token phase}, with a single node $T_0$ having $\isTokenHolder=\mathsf{True}$. 
While all nodes have empty queues ~$Q_v$, they all wait in \Cref{line:cycle token phase - wait for REQ or new Qu message} and thus, if the queues remain empty indefinitely, we have~$t_1=\infty$.

Otherwise, at some time there is at least one node~$u$ that enqueues to~$Q_u$ a message to %
be simulated.
Each such~$u$ sends a $\REQ$ pulse in \Cref{line:cycle token phase - send REQ} and waits to receive a $\REQ$ pulse in \Cref{line:cycle token phase - wait for REQ after seding one}, unless it has already received such a pulse in \Cref{line:cycle token phase - wait for REQ or new Qu message}. 
As there is at least one such node, at least one $\REQ$ pulse is sent. The rest of the nodes first wait to receive a $\REQ$ pulse and then forward it. It follows that all nodes eventually receive and send a single $\REQ$~pulse.
Let us denote by $\tilde P_u$ the partial transcript of a node $u$ at the ``current'' time (which evolves with the proof), then $\forall u\in V$, we have $\tilde P_u =1$ after sending the $\REQ$. To prove Property~(3), we keep track of the partial transcript $\tilde P_u$, recording the pulses sent by each node.

After sending and receiving a $\REQ$ pulse, any node $u\ne T_0$ waits to receive another pulse (\Cref{token-phase-wait-after-REQ}).
The node~$T_0$ sets $\isTokenHolder_{T_0}=\mathsf{False}$, sends a counterclockwise $\TOKEN$ pulse (\Cref{line:cycle token phase - send TOKEN}) and then waits for another pulse like all other nodes. 
The $\TOKEN$ pulse triggered by $T_0$ propagates counterclockwise until it reaches a node~$v$ with a non-empty~$Q_v$, which must exist. The node~$v$ subsequently sets $\isTokenHolder_{v}$ to $\mathsf{True}$ (\Cref{token-phase-set-tokenholder}). 
Thus, by the above definitions, we get that $T_1=v$ 
and $s_1$~is the time when $v$~executes \Cref{token-phase-set-tokenholder}.
Note that~$T_1$ might get the $\TOKEN$ pulse before getting a $\REQ$ pulse, in which case it delays its actions until a $\REQ$ pulse is received. This has no effect on the proof.

In case $T_1 \ne T_0$, at time $s_1$, all the nodes on $[T_0 \rcurvearrowright T_1)$ have sent a $\TOKEN$ pulse and are now waiting for a $\DATA$ pulse (\Cref{token-phase-wait-after-REQ}) that would switch them to their data phase of epoch $k=1$.  
The nodes on 
$(T_0 \curvearrowright T_1)$ could be in two possible stages: either they are still waiting for a $\REQ$ pulse (Lines~\ref{line:cycle token phase - wait for REQ or new Qu message} or~\ref{line:cycle token phase - wait for REQ after seding one}) as described above, or they are waiting in \Cref{token-phase-wait-after-REQ}. 
Hence at time $s_1$, every node $u \in [T_0 \rcurvearrowright T_1)$ has $\tilde P_u =10$, while every node $u \in (T_0 \curvearrowright T_1]$ has $\tilde P_u =\epsilon$ if it has not yet sent a $\REQ$ pulse, or $\tilde P_u =1$ otherwise. Eventually, perhaps at a different time per node, each node~$u$ thus reaches the partial transcript $\tilde P_u=10^{b_1^u}$ with $b_1^u\in\{0,1\}$ being the indicator of whether $u$ has sent a $\TOKEN$ pulse (namely, whether $u \in [T_0 \rcurvearrowright T_1)$).

In the special case where $T_1=T_0$, at time $s_1$, the token has just reached back at $T_1$; all nodes have sent a $\TOKEN$ pulse, and all nodes but $T_1$ are now waiting for a $\DATA$ pulse (\Cref{token-phase-wait-after-REQ}) that would switch them to their data phase of epoch $k=1$.
Hence at time $s_1$, every node~$u$ has the partial transcript $\tilde P_u=10^{b_1^u}$ with $b_1^u=1$ indicating that $u$ has sent a $\TOKEN$ pulse.

When the node~$T_1$ switches to the data phase (\Cref{alg:cycle compiler data phase}), its queue $Q_{T_1}$ is non-empty and so it sends $d\ge 1$ clockwise $\DATA$ pulses (\Cref{line:cycle data phase - initiate data}). We define $d_1=d$.
These $\DATA$ pulses propagate clockwise through all nodes, after the first received $\DATA$~pulse in each node but $T_1$ triggers it to switch to its data phase, after it has previously received a $\REQ$ pulse.
Note that each such node must have received a $\REQ$ pulse before it receives the first $\DATA$~pulse. This is because its counterclockwise neighbor, who sends the $\DATA$ pulse, moves to the data phase only after it has sent a clockwise $\REQ$ pulse.

Once a node $u \neq T_1$ is in its data phase of epoch~$k=1$, it records each received $\DATA$ pulse. The node  propagates the pulse clockwise and eventually the pulse arrives back at $T_1$. Thus, since $T_1$ sends $d_1$ $\DATA$ pulses, after propagating them,
each node~$u$ has $\tilde P_u = 10^{b_1^u}1^{d_1}$.
Once the $d_1$~clockwise pulses reach back at~$T_1$, and only then, it issues a counterclockwise $\END$~pulse (\Cref{line:cycle data phase - initiate end}).
$T_1$ does not generate nor does it propagate any additional pulses before receiving the propagated $\END$ from the other side of the cycle.
This implies that any $u\ne T_1$ receives exactly $d_1$ clockwise $\DATA$ pulses followed by a counterclockwise $\END$ pulse.
Upon receiving the $\END$ pulse, $u$ processes the message $1^{d_1}$ (\Crefrange{line:cycle data phase - record message}{line:cycle data phase - complete process message}) and forwards the $\END$ pulse (\Cref{line:cycle data phase - forward END}). It then completes its $k$-th data phase and its $k$-th epoch, with $\tilde P_u = 10^{b_1^u}1^{d_1}0$.
The node $T_1$ also has $\tilde P_{T_1} = 10^{b_1^{T_1}}1^{d_1}0$ when it receives the $\END$ pulse and switches to the next token phase, at time $t_1$. At that time, all the other nodes have already processed the $\END$~pulse.
This proves the first part of Property~(1).

Next, we need to prove that at time~$t_1$, none of the nodes has passed \Cref{token-phase-wait-after-REQ}. 
In order for a node to pass \Cref{token-phase-wait-after-REQ}, it must be the case that \emph{after} the node has switched to the token phase, it has received a $\REQ$ pulse followed by one additional pulse (in any direction). We argue this cannot happen.
Indeed, at the time where some node $v$ receives the $\END$ pulse and switches to its second token phase,  only the nodes in ${(T_1 \rcurvearrowright v]}$ have received the $\END$ pulse and only these nodes
have switched to their (second) token phase. In their token phase, they may or may not have sent a clockwise $\REQ$ pulse by this time. 
Thus, only nodes in $[T_1 \rcurvearrowright v)$ might have \emph{received} a $\REQ$ pulse. 
However, it is impossible that they received an additional pulse by time~$t_1$, as we next show. 
Each node in $(T_1 \rcurvearrowright v)$ that has received a $\REQ$ pulse is waiting to receive another pulse (\Cref{token-phase-wait-after-REQ})
and is not generating any pulse.
The node $T_1$, if receiving a $\REQ$ pulse, does not process it and does not send a $\REQ$ pulse before receiving an $\END$ pulse in \Cref{line:cycle data phase - wait for end}. 
It also never sends a pulse in the counterclockwise direction before receiving its $\END$ pulse back.
Furthermore, each node in $(T_1 \curvearrowright v)$ (for $v \ne T_1$) is still executing \Cref{line:cycle data phase - forward DATA}, 
so it only forwards pulses and never generates pulses. 
Finally, $v$ has just started its token phase and is waiting to receive a~$\REQ$.
We conclude that no additional pulse (beyond the $\REQ$ pulse, if sent) can be received by the nodes in $(T_1 \rcurvearrowright v)$. This holds for any $v$ at the time it transitions to its second token phase. It thus holds for all nodes at time~$t_1$, when the $\END$ pulse eventually reaches back at~$T_1$.

Next we prove Property~(2) based on the above description of the first epoch. At the onset (at time~$t_0$), $T_0$ is the only node with $\isTokenHolder_{T_0}=\mathsf{True}$. As mentioned above, $T_0$ sets $\isTokenHolder_{T_0}=\mathsf{False}$ and sends a $\TOKEN$ pulse during its token phase. The propagated $\TOKEN$ pulse is the one that triggers $T_1$ to set $\isTokenHolder_{T_1}=\mathsf{True}$ later, at time~$s_1$. Thus, it is clear that $s_1>t_0=0$, and that $T_0$ releases the token before $s_1$ and $T_1$ becomes a token holder at~$s_1$. 
Later, at time $t_1$, the node $T_1$ completes the first epoch, hence, $t_1>s_1$. The node $T_1$ does not set $\isTokenHolder=\mathsf{False}$ during the time frame $[s_1,t_1]$, and it remains to show that it is the only token holder throughout this time frame. 

It is clear that no node in~$(T_0 \rcurvearrowright T_1)$ has set itself as a token holder as otherwise, that node would have been the node we indicate as~$T_1$. 
After time $s_1$, no more $\TOKEN$ pulses are sent in the first token phase. Further, $T_1$ sends a clockwise $\DATA$ pulse that transitions all other nodes into their data phase. This implies that no node besides~$T_1$ can execute \Cref{token-phase-set-tokenholder} and set $\isTokenHolder=\mathsf{True}$ in this token phase. 
As for the second token phase, each node that reaches it before $t_1$ does not pass \Cref{token-phase-wait-after-REQ} before time~$t_1$, as we have shown above, thus in particular, it does not receive a $\TOKEN$ pulse and does not reach \Cref{token-phase-set-tokenholder}.  

Finally, we prove Property~(3), that is, that all nodes reach a global consistency regarding the sent message of the first epoch. This follows from the above analysis: As we argued, at the time some node~$u$ completes its first epoch, it holds that $\tilde P_u = 10^{b_1^u}1^{d_1}0$. The part $1^{d_1}$ corresponds to the $d_1$ $\DATA$ pulses initiated by~$T_1$, which form the encoding of the message communicated in this epoch. 
This completes the proof of Property~(3).

\medskip

\textbf{Induction Step.}
To complete the proof, we need to prove the induction step. Most of the above proof holds as is for $k>1$, if we replace $s_1,t_1,T_1$ with $s_k,t_k,T_k$, etc. 

Fix $k>1$ with $t_{k-1}<\infty$ (otherwise, the lemma holds vacuously). We use the induction hypothesis on epoch~${k-1}$. We are allowed to do so since $t_{k-1}<\infty$, which implies that at or after time~$t_{k-2}$ there is at least one non-empty~$Q_v$ and the three properties of the lemma apply to epoch~$k-1$. 
The differences between proving the base case and the step are as follows:

\smallskip

In the case where all nodes have an empty~$Q_v$, it is immediate in the base case that no node ever passes 
\Cref{line:cycle token phase - wait for REQ or new Qu message}; we prove the same happens here. However, all the induction hypothesis gives us is that at $t_{k-1}$ all nodes are waiting either in \Cref{line:cycle token phase - wait for REQ or new Qu message} or~\ref{line:cycle token phase - wait for REQ after seding one},
or~\ref{token-phase-wait-after-REQ}. 
Clearly, nodes cannot be in \Cref{line:cycle token phase - wait for REQ after seding one} since their queue is empty.
We now prove they cannot be in \Cref{token-phase-wait-after-REQ} as well. 

Assume towards a contradiction that $v$ is the first to 
pass \Cref{line:cycle token phase - wait for REQ or new Qu message}
in its $k$-th epoch.
Let $\tilde t$ be the time when $v$ receives the $\END$ pulse of its epoch number~$k-1$. After time~$\tilde t$, the node $v$ completes its epoch and transitions to its $k$-th token phase. Since $Q_v$ is empty, $v$ gets to \Cref{line:cycle token phase - wait for REQ or new Qu message} and awaits there for a $\REQ$~pulse. 
Because $v$ eventually reaches \Cref{token-phase-wait-after-REQ}, it must have received a clockwise pulse from its neighbor~$u$, which caused $v$ to pass \Cref{line:cycle token phase - wait for REQ or new Qu message}.

For $v \ne T_{k-1}$, recall that at time~$\tilde t$, the nodes $[T_{k-1}\curvearrowright v)$ are still in their data phase after receiving $d_{k-1}$ $\DATA$ pulses. Recall also that they do not generate new pulses but only propagate pulses, and they do not propagate any additional clockwise pulses because $T_{k-1}$ does not generate any clockwise pulses until the counterclockwise $\END$ reaches it. 
The above-mentioned neighbor $u$ belongs to $[T_{k-1}\curvearrowright v)$, hence, it does not propagate further clockwise pulses from time~$\tilde t$ until $u$ gets the $\END$ pulse.
In case $v=T_{k-1}$, $u$ gets the $\END$ pulse before time~$\tilde t$.

In any case, after $u$ gets the $\END$ pulse, $u$ transitions to its $k$-th token phase and reaches \Cref{line:cycle token phase - wait for REQ or new Qu message}. Therefore, if $u$ did send a $\REQ$ pulse that causes $v$ to pass \Cref{line:cycle token phase - wait for REQ or new Qu message}, then $u$ would have also passed~\Cref{line:cycle token phase - wait for REQ or new Qu message} prior to sending this $\REQ$~pulse, in contradiction to our assumption that $v$ is the first node to pass \Cref{line:cycle token phase - wait for REQ or new Qu message}.

\smallskip

If some node has a non-empty~$Q_v$, in the base case we had that all nodes begin the token phase at the same time~$t_0$, and then send a $\REQ$ pulse if their queue is non empty or if they receive a $\REQ$ pulse. 
When considering the induction step at time $t_{k-1}$, some nodes may have already started their $k$-th epoch, and have already sent a $\REQ$ pulse before time~$t_{k-1}$, as given by Property~(1) of the induction hypothesis. 
The behavior from this point on remains the same as described above for the base case.

\smallskip

For proving the global consistency property in the induction step, let $P_{u,k-1}$ be the transcript of~$u$ at the end of its epoch~$k-1$. By Property~(3) of the induction hypothesis, we know that there exist $d_1,\ldots,d_{k-1}$ and $b^u_1,\ldots, b^u_{k-1}$ for any $u\in V$ such that $P_{u,k-1} \triangleq 10^{b^u_1}1^{d_1}0\cdot10^{b^u_2}1^{d_2}0\dots10^{b^u_{k-1}}1^{d_{k-1}}0$ for any $u\in V$.
Furthermore, the above analysis shows that there exist an integer $d_k>0$ and an indicator $b^u_k\in\{0,1\}$ per $u$,
such that the pulses sent by node~$u$ during its $k$-th epoch can be described by 
$10^{b_k^u}1^{d_k}0$, where $1^{d_k}$ signifies the $\DATA$ pulses sent by $T_k$, which encodes the communicated message of this epoch. This gives Property~(3).
\end{proof}

Next, 
we show that \Cref{lemma:cycle} implies the correctness of the simulator (\Cref{thm:cycle-correctness}). We then analyze its overhead (\Cref{overhead-lemma}). Finally, we discuss in Section~\ref{sec:binary} %
ways to improve the obtained complexity (\Cref{lem:overhead-binary}).
Together, these three prove \Cref{thm:simpleCycle}.

\begin{theorem}
\label{thm:cycle-correctness}
Let $G=(V,E)$ be a cycle.
For any asynchronous protocol~$\pi$, let $\hat \pi$ be the protocol defined by  \ALGone
with the input~$\pi$. 
Then, executing $\hat \pi$ on the fully-defective~$G$ simulates an execution of $\pi$ on the noiseless network~$G$.
\end{theorem}

\begin{proof} 
Let $\mathcal{E}_{\hat\pi}$ be an execution of $\hat \pi$ over the fully-defective cycle~$G$. 
We derive a transcript $\tau$ from $\mathcal{E}_{\hat\pi}$, and claim that $\tau$ corresponds to a valid transcript of some execution of~$\pi$ in the noiseless network~$G$, which we denote $\mathcal{E}_\pi$. 
The reader should distinguish between the simulated~$\pi$ which is the black-box interface used as an input of
\ALGone,
and the protocol $\pi$~that generates the execution $\mathcal{E}_\pi$ on the noiseless network~$G$. 
In order to avoid confusion, we will use the term \emph{simulated $\pi$} to denote the former and refer to $\mathcal{E}_\pi$ when discussing the latter.

Let us specify the structure of the transcript~$\tau$. We can think about it as an ordered sequence of events $\tau=\tau_1 \tau_2 \cdots$, where $\tau_i$ is either the event that some node $u$ sent a message $m$ to~$v$, i.e., $\tau_i=(\langle\mathsf{send}\rangle,u,v,m)$ or the event that some node $u$ received a message $m$ from $v$, i.e., $\tau_i=(\langle\mathsf{receive}\rangle,u,v,m)$.

To derive $\tau$ from $\mathcal{E}_{\hat\pi}$, we follow the execution of $\hat \pi$ as the time evolves.
We add a \emph{send} event
every time the simulated~$\pi$ instructs node~$u$ to send a new message~$m$. Specifically, when node~$u$ enqueues $M=(m,u,v)$ to~$Q_u$,  
we add the event $(\langle\mathsf{send}\rangle,u,v,m)$ to~$\tau$.
Additionally, every time some node~$v$  delivers the message $M=(m,u,v)$ to the simulated~$\pi$ (\Cref{line:cycle data phase - deliver}),
we add the event
$(\langle\mathsf{receive}\rangle,v,u,m)$ to~$\tau$.
Recall that $\pi$ generates messages to~$u$ sequentially and $\tau$ maintains this order; events that happen at the same time in different nodes are ordered  arbitrarily in~$\tau$.

Given $\mathcal{E}_{\hat\pi}$ and its derived~$\tau$, we prove that there exists an execution $\mathcal{E}_\pi$ of~$\pi$ on the noiseless network~$G$ that produces these exact same events in the same order, i.e., such that $\tau$ is exactly the transcript of~$\mathcal{E}_\pi$.

The execution $\mathcal{E}_\pi$ is obtained by executing~$\pi$ over~$G$ with the following scheduler that imitates the behavior of~$\mathcal{E}_{\hat\pi}$.
Every time a node sends a message in~$\mathcal{E}_\pi$, the message is delayed at the channel and delivered only at the time the respective message is received in~$\mathcal{E}_{\hat\pi}$. That is, our scheduler ``follows'' the execution of $\mathcal{E}_{\hat\pi}$, and delays each message until the time its corresponding message is delivered in $\mathcal{E}_{\hat\pi}$. 
Specifically, whenever a \emph{receive} event is registered in~$\tau$ (in~$\mathcal{E}_{\hat\pi}$), we deliver the corresponding message in~$\mathcal{E}_\pi$.
The scheduler also controls the execution time of all nodes, which enables it to control the timing of the send events $\pi$ initiates in $\mathcal{E}_\pi$ so they correspond to the same order of send events in~$\mathcal{E}_{\hat\pi}$. 
We now show that the above defines a valid scheduler, and that the resulting $\mathcal{E}_\pi$ has the transcript~$\tau$.

Define $\mytime(j)$ to be the time in $\mathcal{E}_{\hat\pi}$ when the event $\tau_j$ is registered (note, for multiple events that occur at the same time, we let $\mytime(j)$ refer only to the events up to~$\tau_j$). 
We prove the following statement by induction on~$j$:
(1) The scheduler is valid: whenever instructed to deliver a message~$m$, this $m$ was issued to the channel and hasn't been delivered yet.
(2) The transcript~$\tau$ derived from $\mathcal{E}_{\hat\pi}$ is a prefix of the transcript of~$\mathcal{E}_\pi$.

For the base case, $\mytime(0)$, the transcript $\tau$ is empty. It is clear that the scheduler is (vacuously) valid, and that $\tau$ is a prefix of the transcript of~$\mathcal{E}_\pi$.

We proceed with the induction step. Assume that the induction statement holds at~$\mytime(j-1)$, that is,
at this point in time, the events $\tau_1\cdots \tau_{j-1}$ are a prefix of the events in~$\mathcal{E}_\pi$, and all the actions of the scheduler so far are valid.
Now consider the next event recorded to~$\tau$. There are two options here, either it is a send event or a receive event. 

In the first case, let 
$\tau_j= (\langle\mathsf{send}\rangle,u,v,m)$. 
Consider $\mathcal{E}_\pi$ right after the event $\tau_{j-1}$, i.e., at $\mytime(j-1)$. 
Every node~$u$ in~$\mathcal{E}_\pi$ has exactly the same state as the simulated $u$ in the simulated~$\pi$, which follows from the induction hypothesis. Therefore, if the simulated $\pi$ instructs $u$ to send the message $m$ to $v$ (which triggers $\tau_j$ in~$\mathcal{E}_{\hat\pi}$),  the same (eventually) happens at node~$u$ in~$\mathcal{E}_\pi$. The scheduler delays all other nodes until the same message is sent in~$\mathcal{E}_\pi$,
and the claim thus holds after event $\tau_j$, i.e., at $\mytime(j)$ as well.

The other case is when the $j$-th event is a receive event, say, $\tau_{j}=(\langle\mathsf{receive}\rangle,u,v,m)$.
Consider the node~$u$ that executes \Cref{line:cycle data phase - deliver} which corresponds to this event.
By Property~(3) of \Cref{lemma:cycle}, this message is sent by the message sender of that epoch,~$v$. Denote this epoch by~$k$. This means that at the beginning of epoch $k$ the message~$m$ appeared in~$Q_v$ and was dequeued by $v$ at the beginning of the data phase of epoch $k$; note that dequeued messages are never enqueued back to~$Q_v$. This means that at some point in time before~$\mytime(j)$, node~$v$ enqueued this message to~$Q_v$ and it was never dequeued before the $k$-th epoch; let $\tau_i$ with $i<j$ be the corresponding event of enqueuing~$m$ to~$Q_v$.
Now consider $\mathcal{E}_\pi$. By the induction hypothesis we know that up till event $\tau_{j-1}$ at~$\mytime(j-1)$, the transcript~$\tau$ describes the execution~$\mathcal{E}_\pi$, thus, the message $m$ was sent at $\mytime(i)$ and is currently being delayed by the channel (since the first and only delivery of~$m$ in $\mathcal{E}_{\hat\pi}$ occurs at~$\mytime(j)$). 
The scheduler then instructs the channel to deliver this message, which is a valid action as this message was already issued to the channel and never delivered before.
This completes the inductive proof.

\medskip

The above proves that at any point in time the execution $\mathcal{E}_{\hat\pi}$ over the fully-defective~$G$ simulates a prefix of a valid execution of $\pi$ over the noiseless~$G$. It remains to show liveness, namely, that the prefix keeps growing. 
This follows from Properties~(1) and~(3) of \Cref{lemma:cycle}: the simulation makes progress as long as some $Q_v$ is non-empty or eventually becomes non-empty. 
Progress means that all nodes begin and complete their next epoch. In each epoch one message (from some~$Q_v$) is being delivered to its destination.
If the simulated~$\pi$ of some node gives an output, the same node will give the same output in~$\mathcal{E}_{\hat\pi}$. 
If we consider the point in time where all nodes have given output, then all these outputs are valid since $\tau$ at that time is a prefix of the execution $\mathcal{E}_{\pi}$, which also gives the same outputs, by definition.

The only situation where the simulation could reach quiescence is when all the queues~$Q_v$ are empty and remain empty indefinitely. 
But $\tau$ up to that time, as argued above, is a transcript of some~$\mathcal{E}_\pi$ on the noiseless~$G$, where no messages are currently delayed by any channel, and no new messages are going to be sent since the nodes in~$\mathcal{E}_\pi$ are at the same state as in the simulated~$\pi$. Thus, $\mathcal{E}_\pi$ has reached quiescence as well. 
\end{proof}

Let us point out a couple of additional remarks about our simulation.

\begin{remark}\textbf{FIFO:}
The scheduler derived from our simulator maintains FIFO: Consider an execution~$\mathcal{E}_{\hat\pi}$ of the simulator~$\hat\pi$. If multiple messages from~$u$ to~$v$ exist in the simulation, they are enqueued and communicated by their order. These enqueues translate in~$\mathcal{E}_\pi$ 
to messages sent over the same link. However, the scheduler for $\mathcal{E}_\pi$ will deliver them according to their order in~$\mathcal{E}_{\hat\pi}$'s transcript, which is their order in~$Q_u$, that maintains a FIFO property. This strengthens our result, that is, the simulation works correctly both with or without FIFO assumptions for the simulated protocol.
\end{remark}

\begin{remark}\textbf{No starvation:}
Our proof shows that as long as some~$u$ has a message to send, then \emph{some message} will be sent during the next epoch.
Since the $\TOKEN$ pulse travels counterclockwise sequentially in the cycle, there can be at most $n-1$ epochs until $u$ becomes the token holder. 
Thus, our simulator actually satisfies the stronger notion of \emph{no starvation}. 
\end{remark}

\begin{remark}\label{rem:broadcast}
\textbf{Broadcast:}
We note that by design, our simulator offers an additional \emph{broadcast} operation. That is, a node can send a message whose destination is all other nodes. 
To provide this functionality, we utilize the fact that every message arrives at all nodes, regardless of its original destination.
To broadcast a message~$m$, a node simply fixes its destination to be~$*$. 
Each node that decodes a message 
delivers it to~$\pi$ if its destination is either that node (as before) or~$*$.
We will use this feature in our Robbins cycle construction in \Cref{section:robbins}.
\end{remark}

\begin{lemma}\label{overhead-lemma} The overhead of simulating a single message $m$ in \ALGone is $\CCm(m)=|V|^{O(1)}\cdot 2^{|m|}$.
\end{lemma}
\begin{proof}
Let $n=|V|$ be the length of the cycle. Suppose the message $M=(m,u,v)$ is dequeued in some epoch and is being communicated (i.e., $m$ is being communicated by the simulated~$\pi$ over the link~$(u,v)$). We can write $|M|=|m|+O(\log n)$.
Communicating~$M$ over the cycle results in the following pulses sent by \emph{each} node during this epoch (Property~(3) of Lemma~\ref{lemma:cycle}):
a single $\REQ$ pulse, at most a single $\TOKEN$ pulses, $2^{|M|}$ $\DATA$ pulses and a single $\END$ pulse. 
Since there are $n$~nodes, where each node sends at most $3+2^{|M|}$ pulses,  we conclude that $\CCm(m)=O(n\cdot2^{|m|+O(\log n)})=n^{O(1)}\cdot 2^{|m|}$.
\end{proof}

\subsection{Reducing the communication via binary encoding}
\label{sec:binary}

Encoding each message via a unary encoding leads to a pulse overhead that is exponential in the message size: $\CCm(m)=\text{poly}(n, 2^{|m|})$, with $n=|V|$.
We now 
show how to send messages over a simple cycle via a binary encoding of the message. 
This binary encoding leads to a much improved communication complexity of $\CCm(m)=O(n\cdot|m|+ n\log n)$.

Let $M=(m,u,v)\in \{0,1\}^*$ be the message that the token holder~$u$ wishes to send.
The idea is to encode the bits of~$M$ so that a clockwise pulse denotes the bit~$1$, and a counterclockwise pulse denotes the bit~$0$; we denote these as $\DATA(1)$ and $\DATA(0)$, respectively.
Since the order of the bits is important, the token holder sends the next bit only after receiving the previous bit from the other direction of the cycle. (As an optimization, all the pulses of consecutive same-bit sequences may be sent concurrently, and then the token holder should wait to receive all the pulses of the same direction before sending pulses in the other direction. For clarity of the presentation, we do not delve into the details.)
However, now that counterclockwise pulses signify a 0 data bit, 
the challenge is that we need a different way to indicate the end of transmitting~$M$, that is, we need a way to encode an~$\END$ pulse. 

We overcome this challenge by having the nodes agree on a fixed parameter~$L$.
In order to communicate that $M$'s transmission has completed (replacing the $\END$ pulse), 
the token holder sends $L\ge2$~consecutive counterclockwise pulses. 
In addition, the bitstring~$M$ is padded with a $1$ after every $L-1$ consecutive $0$s (when read from its first symbol and onward).
An additional trailing $1$ is sent after $pad(M)$ and guarantees that, even if $M$ has ended with a~$0$ or a sequence of~$0$s, then $pad(M)\cdot1\cdot0^L$ has $L$ consecutive~$0$s only at its suffix.   
Furthermore, we add a preceding $1$ before~$pad(M)$: Recall that the token holder must initiate the sending protocol with a clockwise $\DATA$~pulse, as 
otherwise, the sender's first %
counterclockwise
pulse might be mistaken for a $\TOKEN$~pulse in those nodes that have not yet forwarded a $\TOKEN$ pulse and are still in the token phase.
To summarize, in order to communicate the message~$M$,  the token holder communicates pulses according to the encoded message $Z=1\cdot pad(M)\cdot1\cdot0^L$.

The revised data phase algorithm is given in Algorithm~\ref{alg:cycle compiler data phase-binary}.
\begin{algorithm}[ht]
\caption{Data phase: Broadcasting a message, binary version (node $u$)}\label{alg:cycle compiler data phase-binary}
\begin{algorithmic}[1]
\Statex An integral parameter $L\ge 2$ is agreed upon all nodes.
\Statex
\If{$\isTokenHolder_u$} %
    \State dequeue a message from $Q_u$ and denote it as $M=(m,u,v)$ \label{line:BinaryData:dequeue}
    \State \parbox[t]{\columnwidth-\algorithmicindent-\algorithmicindent}{let $pad(M)$ be the string obtained from $M$ by inserting a~$1$ after every $L-1$ consecutive $0$s of~$M$\strut}
    \State $Z \gets 1\cdot pad(M)\cdot 1\cdot 0^L$  
    \label{line:Binary:setZ}
    \For {$j=0$ to $|Z|-1$}
        \If {$Z_j = 1$} 
            \State send a clockwise $\DATA(1)$ pulse 
            \State wait until a clockwise pulse is received
        \Else %
            \State send a counterclockwise $\DATA(0)$ pulse 
            \State wait until a counterclockwise pulse is received
        \EndIf
    \EndFor
\Else %
    \Repeat 
    \State forward every incoming pulse along the cycle, according to its original direction
    \State record each clockwise pulse as a~$1$ and each counterclockwise pulse as a~$0$
    \Until{$L$ consecutive $0$s have been recorded}
	\State let $Z$ be the recorded string. Parse $Z= 1\cdot P\cdot1\cdot0^L$
	\State 
	\parbox[t]{\columnwidth-\algorithmicindent-\algorithmicindent}{%
	let $M'\gets pad^{-1}(P)$ be the string obtained by removing any~$1$ that appears after a sequence of $L-1$ consecutive~$0$s. Parse $M'=(m',u',v')$}
	{\label{line:binary:decode}}
	\IIf{$u = v'$} 
	    deliver $m'$ to $\pi$ (as if received from $u'$) 
	    \label{line:BinaryData:deliverMsg}
	\EndIIf
\EndIf
\State continue with Algorithm~\ref{alg:cycle compiler token phase} 
\label{line:BinaryData:End}
\end{algorithmic}
\end{algorithm}
We argue that replacing Algorithm~\ref{alg:cycle compiler data phase} with Algorithm~\ref{alg:cycle compiler data phase-binary} does not change the premise of Theorem~\ref{thm:cycle-correctness}. 
For the rest of this section, we change \Cref{line:cycle token phase:end} in \Cref{alg:cycle compiler token phase} to say ``continue with \Cref{alg:cycle compiler data phase-binary}''.

\smallskip

We show that an execution of Algorithm~\ref{alg:cycle compiler token phase} along with Algorithm~\ref{alg:cycle compiler data phase-binary} satisfies a \emph{Global consistency} property similar to the one of Lemma~\ref{lemma:cycle}. 
One can easily verify that the \emph{Progress} property and the \emph{Single token holder} property hold as well, with the same proof as before.

\begin{lemma}
\label{lem:binary:adaptation of lemma 3}
Consider the following modification of the \textbf{Global consistency} property: 
\begin{quote}
There exist strings $z_1,\dots, z_k$, where $z_i=1m_i10^L$ for some $m_i \in\{0,1\}^+$, and for any $u\in V$ there are $b^u_1,\dots, b^u_k \in \{0,1\}$, such that when the node~$u$ completes its $k$-th epoch, its sent transcript (the overall pulses sent so far by~$u$) is $P_{u,k} \triangleq
10^{b^u_1}z_1 \cdot10^{b^u_2}z_2\cdot \cdots \cdot10^{b^u_k}z_k$.

In addition, the message each node decodes and processes (\Cref{line:binary:decode}  in \Cref{alg:cycle compiler data phase-binary}) in epoch~$k$ is exactly the message  $pad^{-1}(m_k)$ sent by~$T_k$.
\end{quote}

Then the statement of Lemma~\ref{lemma:cycle} holds for the simulator given by Algorithms~\ref{alg:cycle compiler token phase} and~\ref{alg:cycle compiler data phase-binary}.
\end{lemma}
In order to avoid excessive repetition, we sketch below only the differences from the proof of Lemma~\ref{lemma:cycle} that stem from replacing Algorithm~\ref{alg:cycle compiler data phase} with Algorithm~\ref{alg:cycle compiler data phase-binary}.
\begin{proof}
Recall from the proof of Lemma~\ref{lemma:cycle}, that $T_k$~gains the token during its $k$-th token phase (at time~$s_k$) since its~$Q_{T_k}$ is non-empty and a $\TOKEN$ pulse arrives from its counterclockwise neighbor. The node $T_k$ then switches to its data phase (\Cref{alg:cycle compiler data phase-binary}).

The node $T_k$ dequeues a message~$M$ from its queue~$Q_{T_k}$ and sets $z_k= 1\cdot pad(M) \cdot 10^L$ in \Cref{line:Binary:setZ}. Thus, its first pulse is a clockwise $\DATA(1)$ pulse, which causes every other node $u$ to switch to its data phase and execute the code with $\isTokenHolder_u = \mathsf{False}$, similarly to the case in the proof of \Cref{lemma:cycle}, upon receiving the first~$\DATA$.

Note that $T_k$ transmits bits sequentially and proceeds to the next bit only after the previous pulse is received from its other side of the cycle. That is, once the first $\DATA(1)$ pulse arrives back at $T_k$, it continues to communicating $pad(M)\cdot10^L$, bit after bit.

The padding $pad(M)$ and the trailing $1$ following it guarantee that there exists only a single substring of $L$ consecutive~0s in~$z_k$, which resides at the suffix of~$z_k$.
It follows that all other nodes receive the string~$z_k$: they record the message communicated by~$T_k$ bit by bit, until they see $L$ consecutive 0s. This sequence appears only at the suffix of~$z_k$ and signifies its termination.
We can thus deduce that the message~$Z$ recorded by each node has the structure $Z=1 \cdot P \cdot  1 0^L$ so each node continues to extracting the part $P$ (whose length is unknown beforehand) and decodes $M'=pad^{-1}(P)$ to obtain the correct message~$M'=M$ communicated by~$T_k$ in \Cref{line:binary:decode}. If the node is the recipient of~$M$ it delivers it to its simulated~$\pi$ (\Cref{line:BinaryData:deliverMsg}).
Each node then completes its $k$-th epoch, and transitions to its token phase~$k+1$ (\Cref{alg:cycle compiler token phase}).

Since every received pulse is forwarded along the same direction it was received,  during the $k$-th epoch each node $u$ transmits exactly the sequence of pulses described by $10^{b^u_k}z_k$, and thus its overall sent transcript is  $P_{u,k} = P_{u,k-1} \cdot 10^{b^u_k}z_k$, which has the correct structure using the induction hypothesis.
The rest of the proof follows the one of Lemma~\ref{lemma:cycle} as is.
\end{proof}

\begin{lemma}\label{lem:overhead-binary} The overhead of simulating a single message~$m$ by \Cref{alg:cycle compiler token phase} and \Cref{alg:cycle compiler data phase-binary} over the simple cycle~$G$ is $\CCm(m)=O(n\cdot |m|+n\log n)$.
\end{lemma}
\begin{proof}
Let $n=|V|$ be the length of the cycle~$G$. Suppose the message $M=(m,u,v)$ is dequeued in some epoch and being communicated (i.e., $m$ is being communicated by the simulated~$\pi$ over the link~$(u,v)$). 
Communicating~$M$ over the cycle results in the following pulses sent by \emph{each} node during this epoch (Property~(3) of Lemma~\ref{lem:binary:adaptation of lemma 3}):
a single $\REQ$ pulse, at most a single $\TOKEN$ pulses, and at most~$2+L+\big(1+\frac{1}{L-1}\big)|M|$ $\DATA$ pulses
(a preceding and trailing~$1$s, $L$ trailing 0s, and $|pad(M)|\leq\big(1+\frac{1}{L-1}\big)|M|$ ``content'' pulses). 
Since $L \ge2$ is a constant, each of the $n$~nodes sends~$O(|M|)=O(|m|+\log n)$ pulses.
We conclude that $\CCm(m)=n\cdot O(|M|) = O(n\cdot |m|+ n\log n)$.
\end{proof}

\section{Simulating computations over fully-defective 2-edge connected networks}
\label{sec:GeneralCycle}
In this section we show how to perform resilient computations over any 2-edge-connected fully-defective network, given a Robbins cycle.

\begin{theorem}[A simulator for a Robbins cycle]
\label{thm:RobbinsCycle}
Let $C$ be a Robbins cycle (over~$G$) and let each node know its clockwise and counterclockwise neighbors for each of its occurrences on~$C$.
There exists a noise-resilient simulator 
over the fully-defective~$G$ for any asynchronous protocol~$\pi$.
The simulator features $\CCinit=0$ and
$\CCm(m) = O(|C|\cdot |m|+|C|\log |V|)$ pulses. 
\end{theorem}

Let $G$ be a 2-edge-connected graph,
and assume the nodes are given a Robbins cycle~$C$, namely, a directed cycle that passes through each vertex at least once, and that does not use any edge in both of its directions. (We stress that this assumption is later removed by showing how to construct the Robbins cycle from scratch, in \Cref{section:robbins}.) 
As a node $u$ may appear in $C$ more than once, we denote by~$k_u$ the number of its occurrences on~$C$. 
The initial knowledge of each node $u$ about $C$ is the value of $k_u$ and its clockwise and counterclockwise neighbors along the cycle. That is, every node~$u$ knows the nodes $\Prev_{u,i}$ and $\Next_{u,i}$ for every $0\leq i\leq k_u-1$, such that the nodes along the cycle~$C$ correspond to the $\Prev$ and $\Next$ variables of all nodes in a consistent manner. We refer to the sequence of nodes  between two successive occurrences of $u$ on $C$ (including the ending occurrence of~$u$) as a \emph{segment} $(u \xrightarrow{} \cdots \xrightarrow{} u]$. %

The high level approach for the algorithm is built upon \ALGone of the simple cycle, with pulses forwarded across the Robbins cycle~$C$. 
By way of mimicking the protocol for the simple cycle, $u$ views each of its occurrences along $C$ as a different node along the cycle. Accordingly, when a node~$u$ is the token holder, it has exactly one occurrence on~$C$ which is associated with holding the token, and when we refer to a token holder in this section, we refer to that precise occurrence. When any node~$u$ forwards a pulse in some direction, it forwards it to the node along~$C$ that follows its occurrence that received the pulse. There are several challenges in this generalization.

\smallskip\textit{\textbf{Challenge 1: Edge repetition along~$C$.}}
Perhaps the main challenge is for $u$ to keep track of its occurrences and distinguish between them: it could be that multiple occurrences of $u$ have the same incoming edge. 
Still, this node
needs to be able to associate each pulse it receives with its appropriate segment, even when pulses that belong to different segments arrive from the same neighbor.

For instance, consider the node~$d$ in~\Cref{fig:RobbinsCycle}, and suppose it has just started its data phase and received a clockwise $\DATA$ pulse from node~$c$. This data pulse could have originated at node~$e$ 
and should be forwarded to node~$a$, or it could have originated at node~$a$ 
and should be forwarded to node~$e$. 

To avoid this type of confusion,
each node~$u$ tracks throughout the execution in which of its segments the token is located.
Specifically, the node~$u$ maintains the invariant that $\Prev_{u,i}$ and $\Next_{u,i}$ reflect the previous and next nodes of its occurrence number~$i$, for $0\leq i\leq k_u-1$, in the specific rotation of the cycle that \textbf{starts} with the token segment considered by $u$, i.e., the token always resides within segment 0, (locally) for all nodes.
To achieve this, node $u$ applies a local rotation function
upon receiving information about the token holder, namely, upon receiving a $\TOKEN$ pulse which we show that can be traced correctly to a specific segment.

\smallskip
\textit{\textbf{Challenge 2: Distinguishing between different $\DATA$ pulses.}}
Another challenge that arises is how to distinguish between different $\DATA$ pulses. Recall that in the simulator for the simple cycle, the $d$ $\DATA$ pulses are forwarded concurrently, in the sense that the token holder issues all $d$ $\DATA$ pulses and only then waits to receive them. However, once an edge appears more than once in $C$, its endpoint $u$ needs to tell apart the case in which it receives two different $\DATA$ pulses on that edge from the case in which it receives the same $\DATA$ pulse on that edge but from different segments. This is crucial because $d$ is not known in advance (and in fact the value of $d$ is the exact piece of information that needs to be learned).

We overcome this challenge by making sure that the $\DATA$ pulses get forwarded in a sequential manner as follows. The node-occurrence that is the token holder does not issue all $d$ $\DATA$ pulses, but rather waits to receive $\DATA$ pulse number~$\ell$ from its counterclockwise neighbor before issuing $\DATA$ pulse $\ell+1$, for $1\leq \ell \leq d-1$.
A node $u$ that receives the $\DATA$ pulse for the $i$-{th} time since the last reception of a counterclockwise $\END$ pulse, forwards it to~$\Next_{u,i-1}$ (where the index is taken mod~$k_u$).

\smallskip
\textit{\textbf{Challenge 3: $\REQ$ pulses have no guaranteed structure.}}
While our approach for overcoming Challenges 1 and 2 allows the nodes to have consistent rotations of the cycle and the token segments for streamlining the $\DATA$, $\END$, and $\TOKEN$ pulses, it is insufficient for handling $\REQ$ pulses. The reason for this is that each of the other three types of pulses traverses the cycle sequentially (or partially traverses in case of a $\TOKEN$ pulse), but $\REQ$ pulses could be initiated by different nodes, so that a node that receives a $\REQ$ pulse does not have any particular promise about its origin and hence cannot tell which neighbor to forward this pulse to.

We remedy this uncertainty by having each node disseminate $\REQ$ pulses to all of its clockwise neighbors, regardless of their origin (which is not known to the node). We show that in the case of $\REQ$ pulses, this coarse action satisfies the conditions that are needed in order for the simulator to work correctly, despite its somewhat more aggressive and unstructured nature.

\subsection{Formal description}
The main idea of the simulator, as mentioned above, is to let each node mimic \ALGone while simulating each one of its occurrences on~$C$ as if it were a separate  node on a simple cycle. Nevertheless, some actions are performed by the node and apply for all its occurrences. We expand on this shortly.

In particular, 
each node~$u$ has the internal variables $\isTokenHolder_u$ and~$Q_u$, for holding the token and queuing its simulated messages. These will serve all its occurrences.
Recall that a segment $(u\to\cdots\to u]$ is a sub-path of the cycle~$C$ between two consecutive occurrences of~$u$. 
Each node $u$ holds the variables $\Prev_{u,i}$ and $\Next_{u,i}$ that reflect the previous and next nodes of its occurrence number $i$, for $0\leq i\leq k_u-1$, see Figure~\ref{fig:cycle-rotation}.

\begin{figure}[htb]
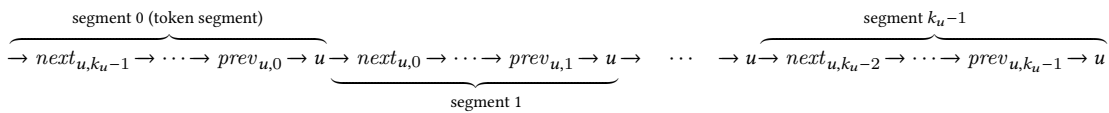

\noindent
\scalebox{0.9}{%
\noindent
\begin{minipage}{\linewidth}  %
\hspace{-1ex}
 \begin{equation*}\!\!
 \overset
    {\text{segment 0 (token segment)}}
    {\overbrace{\xrightarrow{} \Next_{u,k_u-1} \xrightarrow{} \cdots \xrightarrow{} \Prev_{u,0} \xrightarrow{} u}}
 \underset
    {\text{segment 1}}
    {\underbrace{\xrightarrow{} \Next_{u,0} \xrightarrow{} \cdots \xrightarrow{} \Prev_{u,1} \xrightarrow{} u} \xrightarrow{} }
\quad
\cdots
\quad 
 \overset
    {\text{ segment } k_u-1}
    {\xrightarrow{} u 
     \overbrace
        { \xrightarrow{} \Next_{u,k_u-2} \xrightarrow{} \cdots \xrightarrow{} \Prev_{u,k_u-1} \xrightarrow{} u}
    }
 \end{equation*}%
\end{minipage}
}
\caption{The segments of the rotation of~$C$ that starts with the token segment, as seen by a specific node~$u$. The token resides in one of the node-occurrences or links of the token segment.}
\label{fig:cycle-rotation}
\end{figure}

Further, each node~$u$ tracks throughout the execution in which of its segments the token is located and calls this its \emph{token segment} (segment~0). 
The node $u$ applies a local rotation function $\textsc{RotateEdges()}$ upon receiving information about the token holder, namely, upon receiving a $\TOKEN$ pulse, which maintains this invariant.
The procedure $\textsc{RotateEdges()}$ is formally defined as follows.
\smallskip
\begin{mdframed}
\textbf{$\textsc{RotateEdges()}$ for node $u$:} 
Update each $\Prev_{u,i}$ to the previous value of $\Prev_{u,i-1}$ and each $\Next_{u,i}$ to the previous value of $\Next_{u,i-1}$, where indices are taken mod~$k_u$.
\end{mdframed}

\medskip
\noindent
The pseudo-code of our simulator appears in Algorithms \ref{alg:general compiler token phase} and \ref{alg:general compiler data phase} below. 
We are now ready to prove its correctness and analyze its communication complexity.

\renewcommand\thealgorithm{\the\numexpr\value{algorithm}-1}
\renewcommand\thealgorithm{\the\numexpr\value{algorithm}-1(a)}

\begin{algorithm}[htp]
\caption{A simulator for fully-defective networks given a Robbins cycle: token phase (node $u$)}\label{alg:general compiler token phase}
\begin{algorithmic}[1]

\Statex \textbf{Init:} 
A single node has $\isTokenHolder = \mathsf{True}$, associated with one specific occurrence. Node~$u$ holds variables $\Next$ and $\Prev$ for each one of its occurrences, so that these variables (globally) form a Robbins cycle. The first segment (of~$u$ on the cycle) contains the node-occurrence associated with the token. 
Node~$u$ holds a (possibly empty) input $x_u$ for~$\pi$. %
\Statex \textbf{Handling messages sent by $\pi$:}
During the execution of the algorithm, node $u$ enqueues to $Q_u$ any new message $\pi$ asks $u$ to send, in the form (message, source, destination). The actions of~$\pi$ occur in parallel to the execution of this algorithm. 
\Statex

\State \textbf{wait until} $Q_u$ is  not empty or a clockwise $\REQ$ pulse is received from some $\Prev_{u,i}$
\label{line:general token phase:wait on Qu or REQ}
\State send a $\REQ$ pulse to $\Next_{u,i}$ for all $0\leq i \leq k_u-1$\label{line:general token:send reqs}
\State \textbf{wait until} a $\REQ$ pulse is received on each $\Prev_{u,i}$ for all $0\leq i \leq k_u-1$\label{line:general-wait-req1}
\Statex\Comment{Including $\REQ$ pulses received in \Cref{line:general token phase:wait on Qu or REQ}, if any}

\If {$\isTokenHolder_u$}\label{line:g-cycle token phase - release token}
    \State $\isTokenHolder_u \gets \mathsf{False}$ \label{line:g-tokenFalse}
    \State send a counterclockwise $\TOKEN$ pulse to $\Prev_{u,0}$\label{line:g-cycle token phase - send TOKEN}
\EndIf

\Statex
\State \textbf{wait until} receiving a pulse
\label{g-token-phase-wait-after-REQ}
\Statex \Comment{Or process any \emph{second} pulse received in \Cref{line:general-wait-req1} from $\Prev_{u,i}$ for some $0\leq i\leq k_u-1$}
\label{g-token-phase-wait-after-REQ1}
\If {the pulse is a counterclockwise $\TOKEN$ pulse} \Comment{Else, the pulse is a clockwise $\DATA$ pulse}
\label{g-token-phase-wait-after-REQ2}
        \State \textsc{RotateEdges()} \label{line:g-rotate}
        \If {$Q_u$ is not empty}
            \State $\isTokenHolder_u \gets \mathsf{True}$ \label{g-token-phase-set-tokenholder}
        \Else
            \State forward the counterclockwise $\TOKEN$ pulse to $\Prev_{u,0}$ \label{line:g-forward-token}
            \State \textbf{go to} \Cref{g-token-phase-wait-after-REQ}
        \EndIf
    \EndIf\label{line-g-end}
\State continue with \Cref{alg:general compiler data phase} 
\label{line:general cycle token phase:end}
\end{algorithmic}
\end{algorithm}

\renewcommand\thealgorithm{\the\numexpr\value{algorithm}-2}
\renewcommand\thealgorithm{\the\numexpr\value{algorithm}-2(b)}   %

\begin{algorithm}[htp]
\caption{A simulator for fully-defective networks given a Robbins cycle: data phase (node $u$)}\label{alg:general compiler data phase}
\begin{algorithmic}[1]

\setcounter{ALG@line}{\getrefnumber{line:general cycle token phase:end}}

\If{$\isTokenHolder_u$} %
\label{line:g-istoken}
    \State dequeue a message from~$Q_u$, denote it by $(m,u,v)$ and let $1^d$ be its unary encoding \label{line:general data phase - dequeue}
    \For{$d$ times}
        \For{$i$ from $0$ to $k_u-1$}
            \State send a clockwise $\DATA$ pulse to $\Next_{u,i}$\label{line:g-cycle data phase - initiate data} \State \textbf{wait until} receiving a clockwise $\DATA$ pulse from $\Prev_{u,(i+1)\mod k_u}$
        \EndFor
    \EndFor

    \For{$i$ from $k_u-1$ to $0$ } 	\State send a counterclockwise $\END$ pulse to $\Prev_{u,(i+1)\mod k_u}$
	    \State \textbf{wait until} a counterclockwise $\END$ pulse is received from $\Next_{u,i}$\label{line:g-cycle data phase - wait for end}
	\EndFor \label{line:g-endSend}
\Else %
    \Repeat{}\label{line:g-notToken}
            \For{$i$ from $0$ to $k_u-1$}
                \State \textbf{wait until} receiving a clockwise $\DATA$ pulse from $\Prev_{u,i}$ 
                \label{line:general data phase:wait for data:non-TH}
                \Statex \Comment{Including the DATA pulse received in the preceding token phase, if exists}
                \State forward the clockwise $\DATA$ pulse to $\Next_{u,i}$
            \EndFor
    \Until{receiving a counterclockwise $\END$ pulse}\label{line:g-getEnd}

	\State let $count$ be the total number of clockwise $\DATA$ pulses received by $u$ divided by $k_u$ \label{line:g-count}
	\State decode $1^{count}$ as the unary encoding of $(m',u',v')$
	\label{line:general data phase - decode}
	\IIf{$u = v'$} 
	    deliver $m'$ to $\pi$ (as if received from $u'$) \label{line:g-cycle data phase - deliver}
    \EndIIf \label{line:g-endDecode}
    
    \For{$i$ from $k_u-1$ to $0$} \label{line:g-start forward end}
	    \State \textbf{wait until} a counterclockwise $\END$ pulse is received from $\Next_{u,i}$
	    \Statex \Comment{for $i=k_u-1$ this is already received in \Cref{line:g-getEnd}}
    	\State forward the counterclockwise $\END$ pulse to $\Prev_{u,i}$
	\EndFor \label{line:g-end-notToken}
\EndIf
\State continue with Algorithm~\ref{alg:general compiler token phase} 
\end{algorithmic}
\end{algorithm}

\renewcommand\thealgorithm{\the\numexpr\value{algorithm}-1}

\subsection{Analysis}
Similar to the analysis in \Cref{section:analysis-cycle}, we begin by proving the technical \Cref{lemma:cycle-general} that specifies the behavior of the simulation and replaces Lemma~\ref{lemma:cycle}. This technical lemma is then used to argue the correctness of our simulation over a Robbins cycle (\Cref{thm:main-general}). We prove the complexity of our simulation
in \Cref{lemma:g-cost,lemma:g-cost-binary}. Together, these prove \Cref{thm:RobbinsCycle}.

For the analysis, we use the same notations as in \ALGone, up to the following modification. Since we have to be careful and distinguish between the different occurrences of a node on $C$, we let $T_k$ denote the following: Consider the $k$-th node to have set its $\isTokenHolder_u$ to $\mathsf{True}$ in \Cref{g-token-phase-set-tokenholder}. $T_k$ is the \emph{node-occurrence} of this node that has received a $\TOKEN$ and subsequently set $\isTokenHolder_u$ to $\mathsf{True}$.

\begin{lemma}
\label{lemma:cycle-general}
Consider an execution of 
\ALGthree
and consider any ${k \ge 1}$, 
for which ${t_{k-1}<\infty}$.
If from time~$t_{k-1}$ and forward, all queues $\{Q_v\}_{v\in V}$ are always empty, then $t_{k}=\infty$.
Otherwise,
the following hold:

\begin{enumerate}[label = (\arabic*)]

\item \label{g-Progress} \textbf{Progress:} 
All nodes eventually complete their $k$-th epoch. 
In particular, $t_{k}<\infty$, and 
at time~$t_k$, all nodes have already processed the $\END$~pulse (of epoch~$k$) but have not yet passed  \Cref{g-token-phase-wait-after-REQ1} in epoch~$k+1$
(they are either waiting in Lines
\ref{line:general token phase:wait on Qu or REQ} or \ref{line:general-wait-req1} 
for a $\REQ$ pulse, or waiting in  \Cref{g-token-phase-wait-after-REQ1} for either a $\DATA$ or a $\TOKEN$ pulse).

\item \label{g-Single token holder}\textbf{Single Token Holder:}
It holds that $t_{k-1} < s_k < t_k$.
At each moment in~$(t_{k-1},t_k)$, there is at most a single node $u$ for which  $\isTokenHolder_u=\mathsf{True}$ and $u$ associates this with a single occurrence on $C$.
More specifically, within this time frame, the token is passed as follows: the node-occurrence $T_{k-1}$ releases the token at some time in~$(t_{k-1},s_{k})$ and the node-occurrence~$T_k$ is the next node that gains the token at time $s_k$. The node-occurrence $T_k$ (solely) holds the token in~$[s_k,t_k]$.

\item \label{g-Global consistency} \textbf{Global consistency:} 
There exist integers $d_1,\dots, d_k > 0$ and for any $u\in V$ there are $b^u_1,\dots, b^u_k \in \{0,1\}$, such that when the node~$u$ completes its $k$-th epoch, the sent transcript by each of its occurrences is $P_{u,k} \triangleq 10^{b^u_1}1^{d_1}0\cdot10^{b^u_2}1^{d_2}0\dots10^{b^u_k}1^{d_k}0$.

In addition, the message each node decodes and processes (\Cref{line:general data phase - decode}) at its $k$-epoch is the unary decoding of $1^{d_k}$, which is the message sent by $T_k$ (\Cref{line:general data phase - dequeue}). 
\end{enumerate}
\end{lemma}

\begin{proof}
In essence, we wish to follow the line of proof of \Cref{lemma:cycle}. The high-level observation is that in the general case, every \emph{occurrence} of a node on $C$ behaves as in the case of the simple cycle, rather than every node behaving that way. There are a few subtle exceptions to this, which do not harm the proof but are rather essential for allowing it to go through.
We elaborate as follows.

\smallskip\textit{\textbf{Token Phase.}} 
For the token phase, if a node~$u$ has a non-empty queue~$Q_u$,
then it reaches \Cref{line:general token:send reqs} and sends a $\REQ$ pulse to \emph{each} of its clockwise neighbors, $\Next_{u,i}$ for all $0\leq i\leq k_u-1$, and waits in 
\Cref{line:general-wait-req1} to receive a $\REQ$ pulse from \emph{each} of its counterclockwise neighbors, $\Prev_{u,i}$ for all $0\leq i\leq k_u-1$. This is equivalent to saying that every occurrence of~$u$ on~$C$ sends a single clockwise $\REQ$ pulse and 
waits to receive a single counterclockwise $\REQ$ pulse. Thus, Lines~\ref{line:general token phase:wait on Qu or REQ}--\ref{line:general-wait-req1} are equivalent to 
Lines~\ref{line:cycle token phase - wait for REQ or new Qu message}--\ref{line:cycle token phase - wait for REQ after seding one} 
of \ALGone for each occurrence of~$u$.

Similarly, any node~$u$ that receives a $\REQ$ pulse from $\Prev_{u,j}$ for some $0\leq j\leq k_u-1$ in \Cref{line:general token phase:wait on Qu or REQ}, forwards it to each of 
its neighbors $\Next_{u,i}$ for all $0\leq i\leq k_u-1$,
and waits in \Cref{line:general-wait-req1} to receive a $\REQ$ pulse from each of its neighbors $\Prev_{u,i}$ for all $0\leq i\leq k_u-1$ for 
which $i\neq j$. 
For occurrence $j$ of~$u$, this is equivalent to Lines~\ref{line:cycle token phase - wait for REQ or new Qu message}--\ref{line:cycle token phase - wait for REQ after seding one} of \ALGone. 
For other occurrences of~$u$, this is slightly different, as they first forward the $\REQ$ pulse and only then wait to receive it. 
However, this still satisfies that if some $\REQ$ pulse is sent in an epoch, then each occurrence of every node sends and receives exactly one $\REQ$ pulse in that epoch, which is all that is needed for the proof of \Cref{lemma:cycle}. 
Notice that in \Cref{line:general-wait-req1} a node may receive a second clockwise pulse from $\Prev_{u,i}$ for some single $0\leq i\leq k_u-1$, in which case this is a $\DATA$ pulse that is processed in \Cref{g-token-phase-wait-after-REQ1}.

Consider now the node~$u$ which has $\isTokenHolder_u$ set to $\mathsf{True}$. In Lines~\ref{line:g-cycle token phase - release token}-\ref{line:g-cycle token phase - send TOKEN}, this node sends a counterclockwise $\TOKEN$ pulse to $\Prev_{u,0}$. This corresponds to having only occurrence 0 of $u$ on $C$ send a $\TOKEN$ pulse, which is the same as Lines~\ref{line:cycle token phase - release token}--\ref{line:cycle token phase - send TOKEN} in \ALGone and indeed the proof of \Cref{lemma:cycle} needs that only a single $\TOKEN$ pulse traverses the cycle.

It remains to show that each occurrence of a node~$u$ on~$C$ that receives a $\TOKEN$ pulse forwards it to its counterclockwise neighbor in $C$ if the queue~$Q_u$ is empty, or sets $\isTokenHolder_u$ to $\mathsf{True}$ otherwise, and that there is exactly one occurrence of one node among those with a non-empty queue which receives a $\TOKEN$ pulse. For this, we need to show that each node correctly keeps track of its token segment. We rely on the local $\textsc{RotateEdges()}$ procedure by showing that the following invariant holds in any time throughout the execution: 
Let $v*$ be the node-occurrence of $v$ that is associated with $v$ having $\isTokenHolder_v$ set to $\mathsf{True}$, or the node-occurrence that most recently received a $\TOKEN$ pulse if no such $v$ exists. Then for every node $u$, the occurrence $v*$ is located in segment $0$ of $u$ on $C$ (i.e., $v*$ is located in  $[\Next_{u,k_u-1},\dots,\Prev_{u,0},u]$).
This invariant is assumed to hold at the onset of the execution. Since the invariant holds, once a $\TOKEN$ pulse reaches a node $u$ in \Cref{g-token-phase-wait-after-REQ2}, it must reach it from $\Next_{u,k_u-1}$. Then, $u$ invokes $\textsc{RotateEdges()}$ in \Cref{line:g-rotate} before setting $\isTokenHolder_u$ to $\mathsf{True}$ in \Cref{g-token-phase-set-tokenholder} or forwarding the $\TOKEN$ pulse in \Cref{line:g-forward-token}, depending on whether its queue $Q_u$ is empty. In either case, the invocation of $\textsc{RotateEdges()}$ guarantees the invariant is maintained. 

Now, consider the case where a node $v$ sets its $\isTokenHolder_v$ to $\mathsf{True}$. 
For the node-occurrence~$v*$,  Lines~\ref{g-token-phase-wait-after-REQ}-\ref{line-g-end} correspond to Lines~\ref{token-phase-wait-after-REQ}-\ref{line:c-end} in \ALGone, that is, $v*$ receives a $\TOKEN$ pulse and switches to its data phase.
For the other occurrences of~$v$ this is slightly different, as the node~$v$ along with all its occurrences, switches to its data phase once $v*$ obtains the token and node~$v$ executes \Cref{line:general cycle token phase:end} after completing \Cref{g-token-phase-set-tokenholder}. 
This is fine, since all the occurrences at this point have sent and received a $\REQ$ pulse, and can switch to the data phase. Indeed, some of $v$'s node-occurrences might have already forwarded a $\TOKEN$ pulse before and remained in the token phase (e.g., if $Q_v$ was empty once the token had reached them); these occurrences switch to the data phase ``late'' in comparison to \ALGone. Additionally, some of $v$'s node-occurrences might have not received any pulse in this token phase and they switch ``early'' compared to respective node in \ALGone (i.e., before they receive the first $\DATA$ pulse).
Nevertheless, they are all in the data phase when they need to send and receive $\DATA$ pulses. This essentially corresponds to \Crefrange{token-phase-wait-after-REQ}{line:cycle token phase:end} in \ALGone. 
The same reasoning applies to every other node: 
once one of its occurrences receives the $\DATA$ pulse originated at $v*$ and switches to the data phase, then all its occurrences do so at the same time, but they all wait for the first $\DATA$ pulse to arrive (\Cref{line:general data phase:wait for data:non-TH}) and thus behave similarly to \ALGone, despite the ``early'' transition to the data phase.

\smallskip\textit{\textbf{Data Phase.} For the data phase, if node $u$ has its $\isTokenHolder_u$ set to $\mathsf{True}$, then in \Cref{line:g-istoken} it dequeues a message from $Q_u$ and denotes by $1^d$ its unary encoding. Next, in Lines~\ref{line:general data phase - dequeue}-\ref{line:g-endSend}, for each of the $d$ pulses of $\DATA$ that need to be forwarded, each occurrence of $u$ according to their order on $C$ receives and sends the clockwise $\DATA$ pulse and then receives and sends the counterclockwise $\END$ pulse. This corresponds to Lines~\ref{line:cycle data phase - dequeue}-\ref{line:cycle data phase - wait for end}} in \ALGone, with two subtleties. 

The main subtlety is that \Cref{line:cycle data phase - dequeue} is invoked only once by~$u$, which corresponds to \Cref{line:cycle data phase - dequeue} in \ALGone being invoked only by occurrence $0$ of $u$ in its rotation of $C$. %
This is essential, as otherwise if each occurrence of $u$ initiates $d$ separate $\DATA$ pulses then clearly there will be too many in the system and the message will not be correctly decoded. We emphasize that the queue $Q_u$ is a single queue used by all occurrences of $u$, and hence once a message is dequeued from $Q_u$, other occurrences cannot dequeue it again later in further epochs if they become the occurrence of $u$ that is associated with the $\isTokenHolder_u$ variable being set to $\mathsf{True}$.

The second subtlety is that each $\DATA$ pulse begins its traversal over $C$ only after the previous one is received back at occurrence $0$ of $u$. The latter does not harm the proof as it is only a stronger requirement compared to \ALGone.

Similarly, for every node $u$ whose $\isTokenHolder_u$ variable is set to $\mathsf{False}$, each of its occurrences according to their order on $C$ receives and sends the clockwise $\DATA$ pulse (in \Crefrange{line:g-notToken}{line:g-getEnd}) and the counterclockwise $\END$ pulse (in \Crefrange{line:g-start forward end}{line:g-end-notToken}), for $d$ times. This corresponds to \Cref{line:cycle data phase - forward DATA,line:cycle data phase - forward END} in \ALGone. 

Finally, Lines~\ref{line:g-count}-\ref{line:g-endDecode} are also invoked by any node $u$ only once, in order to avoid delivering to $\pi$ duplicates of the received message, corresponding to Lines~\ref{line:cycle data phase - record message}-\ref{line:cycle data phase - complete process message} in \ALGone. Note that this also means that every node $u$ moves to the token phase of the next epoch in \Cref{alg:general compiler token phase} only after all of its occurrences finish the current data phase in \Cref{alg:general compiler data phase}.

\medskip

The above establishes that we can now repeat the proof of \Cref{lemma:cycle} for obtaining a proof of \Cref{lemma:cycle-general}.
\end{proof}

\Cref{lemma:cycle-general} allows proving the correctness of our simulation, %
as follows.

\begin{theorem}
\label{thm:main-general}
Let $G=(V,E)$ be some graph and let $C$ be a Robbins cycle in it.
Given any asynchronous protocol $\pi$, let $\hat\pi$ be the \ALGthree given the input~$\pi$.
Then, executing~$\hat\pi$ on the fully-defective network~$G$ simulates an execution of~$\pi$ on the noiseless~$G$.
\end{theorem}

\begin{proof}
The proof of \Cref{thm:main-general} is exactly the same as the proof of \Cref{thm:cycle-correctness}, with the modifications that (i) it uses \Cref{lemma:cycle-general} instead of \Cref{lemma:cycle}, (ii) instead of referring to a node in the network, it refers to its occurences along $C$, and (iii) it adjusts the line numbers that reflect the delivery of a message to the protocol~$\pi$ 
(\Cref{line:g-cycle data phase - deliver}
in \Cref{alg:general compiler data phase} instead of Line~\ref{line:cycle data phase - deliver}
in \Cref{alg:cycle compiler data phase}).
\end{proof}

Finally, the following lemma states the message overhead of our simulator. Its proof is identical to that of~\Cref{overhead-lemma}, except that we consider pulses sent by each of the \emph{occurrences of nodes} on $C$, whose length~$|C|$ can be greater than the number of nodes~$n$.
\begin{lemma}
\label{lemma:g-cost}
Given a Robbins cycle~$C$,
the overhead of simulating a single message $m$ in \ALGthree is $\CCm(m)=O(|C|\cdot2^{|m|+O(\log n)})$.
\end{lemma}

A direct application of the binary encoding described in Section~\ref{sec:binary} yields the following optimization.
\begin{lemma}
\label{lemma:g-cost-binary}
Given a Robbins cycle~$C$,
the overhead of simulating a single message $m$ in \ALGthree, replacing the unary encoding with a binary encoding, is $\CCm(m)=O(|C|\cdot |m|+|C|\log n)$.
\end{lemma}
We omit the details as they repeat the proofs in Section~\ref{sec:binary} for reducing the communication complexity in the simple cycle.

\section{Constructing a Robbins cycle in a fully-defective 2-edge connected network}
\label{section:robbins}
\label{app:robbinsConst}

The simulator of Section~\ref{sec:GeneralCycle}
assumes the nodes are given a Robbins cycle on which they communicate. 
In this section, we show how the nodes can construct such a cycle on any  2-edge-connected fully-defective network~$G$. 

Whitney~\cite{Whitney32} proved that any 2-edge-connected graph $G$ can be decomposed into 
\begin{equation*}
G=C_0 \cup E_0 \cup E_1 \cup \cdots \cup E_k,
\end{equation*}
where $C_0$ is a simple cycle, and for any $i\ge0$, $E_i$ is an \emph{ear}---a simple path or cycle whose endpoints belong to $C_0 \cup E_0 \cup \cdots \cup E_{i-1}$. 
Moreover, the process of decomposing $G$ into ears can be performed by starting from a single node, and constructing $C_0$ and $E_i$ in an increasing order $i=0,1,2,...$. 
See also~\cite{lovasz85,KR00,tsin04,schmidt13} for further details and several distributed ear-decomposition algorithms in noiseless settings.
Our Robbins cycle construction essentially performs a distributed
and content-oblivious 
version of Whitney's ear-decomposition process (see, e.g., Lemma~2.1 in~\cite{ramachandran93} for the centralized algorithm), where nodes form the cycle~$C_0$ and the ears~$E_0,E_1,\ldots$ sequentially. 
A newly constructed ear is incorporated with the previous constructions to form a non-simple cycle that includes them all.

We start at a designated root node and perform a content-oblivious DFS by sending a token over edges in a sequential manner;
see, e.g., \cite[Section~5.4]{peleg00}.
This process continues until the token returns to the root which signifies that a cycle is closed \emph{at the root}. 
We require the constructed cycle to be \emph{simple}. Indeed, if the token reaches some node $v\ne root$ twice, then $v$ sends the token back to where it came from, which is equivalent to \emph{backtracking} in the standard DFS algorithm. %
Backtracked edges do not participate in the constructed cycle, and they are left for future ears.

We denote the {simple} cycle constructed by the above procedure by~$C_0$. The order the DFS-token progresses along~$C_0$ defines  the clockwise direction on the cycle. 
Nodes on $C_0$ employ \ALGthree to communicate over~$C_0$ in a noise-resilient manner with the root being the first token holder.

Recall that a directed cycle can be represented by the nodes either locally, i.e., each node knows its clockwise and counterclockwise neighbor(s), or  globally, i.e., knowing the sequence of IDs that defines the cycle.
Our algorithm will use both representations, however, this is done only to simplify the analysis and reduce the length of the constructed Robbins cycle. 
In Remark~\ref{rem:localCycle} we sketch how to remove this assumption.

Before the nodes on $C_0$ continue with adding ears to~$C_0$, they first broadcast their IDs and achieve a global representation of the cycle. The root sends its ID to its neighbor, who appends its own ID and transfers the message to its next neighbor and so on.
When the message reaches the root again, it contains the sequence of IDs of the cycle~$C_0=(root, v_1, v_2, \ldots, root)$. The root broadcasts this information; it will be used towards continuing the Robbins cycle construction.

Next, the nodes on $C_0$ select a new root, denoted by $root_0$, to be one of the nodes on $C_0$ that still has \emph{unexplored edges}, which are edges that do not participate in $C_0$. 
The construction proceeds by constructing a new ear, $E_0$, starting from~$root_0$. Again, the nodes perform a sequential DFS by sending a DFS-token over unexplored edges, until the DFS-token reaches some node~$z_0$ that belongs to~$C_0$. As before, we require the path of the DFS-token to be simple, and backtrack whenever the token reaches twice the same node that does not lie on~$C_0$.

The simple path that the DFS-token has undergone from $root_0$ to~$z_0$, excluding edges that have backtracked in the DFS search, becomes the newly constructed ear~$E_0$.
A new ear can be a simple cycle if it is a closed ear with $z_0=root_0$, or it can be a simple path if it is an open ear with $z_0\ne root_0$.

Based on $C_0$ and the ear~$E_0$, we define a new cycle~$C_1$ that contains all the edges of~$C_0$ and of~$E_0$, possibly multiple times, so that $C_1$ is a closed (non-simple) cycle. Recall that in a Robbins cycle, each edge has a unique orientation and the cycle is not allowed to cross the same edge in both directions.
Thus, we let $C_1$ be the cycle 
\begin{equation*}
root_0
\xrightarrow[C_0]{} root_0 \xrightarrow[E_0]{} z_0 \xRightarrow[C_0]{} root_0\text{.}
\end{equation*}
The notation $a \xrightarrow[P]{} b$ here means that we take the complete path~$P$.
The notation $a \xRightarrow[P]{} b  $ means the
\emph{shortest path} from $a$ to~$b$ implied by the clockwise orientation of edges in~$P$. If multiple such paths exist, we take the first one by lexicographic order. 
Note that this path might not be a sub-path of~$P$,\footnote{For instance, let $P=(a\to b\to c \to b \to e)$, then $a \xRightarrow[P]{} e$ is the path $(a\to b\to e)$ which is not a sub-path of~$P$.} 
however, it is uniquely defined and can be retrieved by any node that holds the sequence of IDs that defines~$P$.

It follows that the paths $root_0
\xrightarrow[C_0]{} root_0$ and $z_0 \xRightarrow[C_0]{} root_0$ are  well defined and known by all nodes on~$C_0$, since all these nodes know the sequence of IDs that lie on~$C_0$, in their respective order.
However, the nodes still need to know 
the IDs on
$root_0 \xrightarrow[E_0]{} z_0$ in order to obtain the sequence of IDs in the new cycle~$C_1$. 
Towards this end we do the following. 

The nodes on $P_0 \triangleq (z_0 \xRightarrow[C_0]{} root_0)$  along with the nodes on~$E_0$ form a simple cycle~$E_0 \| P_0$ (recall that $\|$ denotes concatenation). 
This cycle is locally defined: the nodes on~$E_0$ define their neighbors when they first obtain the DFS-token. Each node on~$P_0$ belongs to~$C_0$ and, as argued above, can locally define its neighbors on~$P_0$.
Then, the root starts communicating over this cycle %
using \ALGthree.
As before, the first thing the nodes do is communicating their IDs. In fact, only the new nodes that are on $E_0$ but not in~$C_0$ need to broadcast their IDs in their respective order, similarly to the way it was done after the completion of~$C_0$. After this part, $root_0$ can simply construct the string of IDs of the nodes in~$C_1$ and communicate it over~$E_0 \| P_0$. 

Then, the root communicates the sequence of IDs of~$C_1$ to all the nodes in cycle~$C_0$. 
That is, the root and the nodes on $P_0$ stop communicating on the cycle $P_0\|E_0$ and switch back to communicating over the cycle~$C_0$.
Next, the root sends a message to instruct all the nodes in $C_0$ to switch to the new cycle~$C_1$. Note that this message need not reach the nodes in~$E_0$, as they are already ``set'' to the correct~$C_1$. 
Since the other nodes are set to communicate over~$C_0$, the nodes in $E_0$ are excluded from this communication and these nodes remain idle until the first pulse arrives, which happens once the rest of the nodes switch to communicate over~$C_1$.

The process then repeats:
for any $i>0$,
$root_i$ is selected to be a node on $C_{i}$ that still has  edges that do not belong to~$C_i$.
The nodes construct a new ear~$E_i$ whose endpoints, $root_i$ and~$z_i$,
belong to~$C_i$. The nodes then locally define the non-simple cycle 
$C_{i+1} = root_i
\xrightarrow[C_i]{} root_i \xrightarrow[E_i]{} z_i \xRightarrow[C_i]{} root_i$, and start communicating over it. Next, the nodes globally learn the sequence of IDs included in $C_{i+1}$, which is required for the next iteration, and so on.  
This process ends when the cycle $C_{i+1}$ contains all the edges of~$G$.
See Figure~\ref{fig:robbins} for a demonstration.

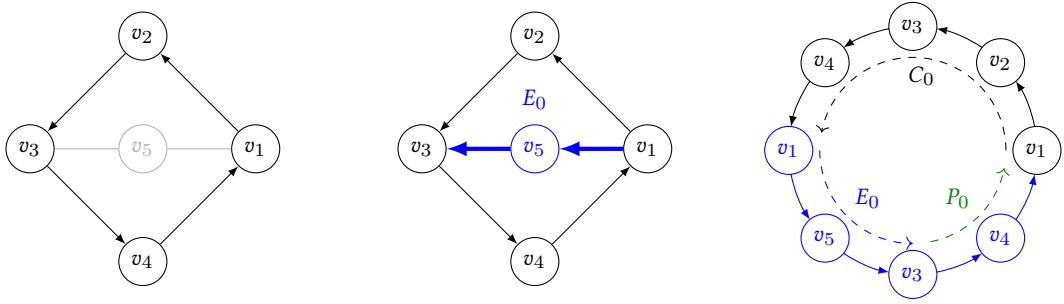
\captionsetup[subfigure]{singlelinecheck=false}
\begin{figure}[htp]
    \centering
    \subfloat[An example of constructing a simple cycle,  $C_0=(v_1,v_2,v_3,v_4)$, starting from the root~$v_1$.  The clockwise direction corresponds to the pulse propagation in Algorithm~\ref{alg:ED_Cycle} and is marked with arrows. The nodes set $\Next_{v_i}=v_{i+1}$ and $\Prev_{v_i} = v_{i-1}$ for $i=1,2,3,4$, except for $\Next_{v_4}=v_1$ and $\Prev_{v_1}=v_4$. 
    Node~$v_5$ does not receive any pulse and is not on~$C_0$. Node~$v_5$ keeps executing Algorithm~\ref{alg:ED_Cycle}, while the rest of the nodes continue to Algorithm~\ref{alg:ED_Ear}.]
    {

\begin{tikzpicture}

\def \n {4}
\def \radius {1.5cm}
\def \margin {10} 

\useasboundingbox  (-\radius*1.5,-\radius*1.5) rectangle (\radius*1.5,\radius*1.5);

\foreach \s/\nxt in {1/2,2/3,3/4,4/1}
{
  \node[draw, circle] (v\s) at ({360/\n * (\s - 1)}:\radius) {$v_{\s}$};

}
\foreach \s/\nxt in {1/2,2/3,3/4,4/1}
{
  \draw[->, >=latex] (v\s) -- (v\nxt);
}

\node[draw,circle,black!30] (v5) at (0,0) {$v_5$};
\draw[-,black!30] (v1) -- (v5) -- (v3);
\end{tikzpicture}} 
\qquad
    \subfloat[Adding the ear $E_0=(v_1 \to v_5 \to v_3)$ to $C_0$. The nodes on $C_0$ execute $\Pi_{\mathsf{NextRoot}}$ to choose as a new root a node that has unexplored edges. Suppose $v_1$ is elected as the new root. 
    The thick blue arrows describe the pulse propagation during the construction of~$E_0$. 
    The nodes $v_1,v_5,v_3$ update their $\Next$ and $\Prev$ to indicate this path.
    Then, nodes on $P_0=(v_3 \to v_4 \to v_1)$, i.e., on ${v_3 \xRightarrow[C_0]{} v_1}$ update their $\Prev$ and $\Next$ to (locally) form the simple cycle $E_0 \| P_0$. 
    This cycle is used to learn the IDs of nodes in~$E_0$. The node~$v_1$ constructs a global representation of~$C_1$ and broadcasts it over~$C_0$ and over~$E_0\| P_0$.]
    {

\begin{tikzpicture}

\def \n {4}
\def \radius {1.5cm}
\def \margin {10} 

\useasboundingbox  (-\radius*1.5,-\radius*1.5) rectangle (\radius*1.5,\radius*1.5);

\foreach \s in {1,...,\n}
{
  \node[draw, circle] (v\s) at ({360/\n * (\s - 1)}:\radius) {$v_{\s}$};
}
\foreach \s/\nxt in {1/2,2/3,3/4,4/1}
{
  \draw[->, >=latex] (v\s) -- (v\nxt);
}

\node[draw,circle,blue] (v5) at (0,0) {$v_5$};
\draw[->,>=latex,blue,ultra thick] (v1) -- (v5);
\draw[->,>=latex,blue,ultra thick] (v5) -- (v3);

\node[blue] at (v5) [above,inner sep=15pt] {$E_0$};
\end{tikzpicture}}
\qquad
    \subfloat[%
    The resulting Robbins cycle~$C_{1}$, obtained after adding the ear $E_0$ described in part~(b) to~$C_0$. 
    The clockwise direction of~$C_1$ is marked with arrows.]
    {

\begin{tikzpicture}[every label/.style={font=\footnotesize,inner sep=15pt}]

\def \n {8}
\def \radius {1.65cm}
\def \innerfrac {0.75}
\def \margin {11} 

\useasboundingbox  (-\radius*1.25,-\radius*1.35) rectangle (\radius*1.25,\radius*1.3);

\foreach \s/\sname in {1/1 ,2/2 ,3/3 ,4/4 }
{
  \node[draw, circle] (v\s) at ({360/\n * (\s - 1)}:\radius) {$v_{\sname}$};
  \draw[->, >=latex] ({360/\n * (\s - 1)+\margin}:\radius) 
    arc ({360/\n * (\s - 1)+\margin}:{360/\n * (\s)-\margin}:\radius);
}

\foreach \s/\sname in {5/1 ,6/5 ,7/3, 8/4 }
{
  \node[draw, circle,blue] (v\s) at ({360/\n * (\s - 1)}:\radius) {$v_{\sname}$};
  \draw[->, >=latex,blue] ({360/\n * (\s - 1)+\margin}:\radius) 
    arc ({360/\n * (\s - 1)+\margin}:{360/\n * (\s)-\margin}:\radius);
}


\draw[->,dashed] (0:\innerfrac*\radius)  arc [start angle=0, end angle={360/ \n *(5-1)-\margin}, radius=\innerfrac*\radius] node [midway,below]{$C_0$};
\draw[->,dashed,blue] ({360/\n * (5-1)}:\innerfrac*\radius)  arc [start angle=360/\n * (5-1), end angle={360/ \n *(7-1)}, radius=\innerfrac*\radius] node [midway,above right]{$E_0$};
\draw[->,dashed,green!50!black] ({360/\n * (7-1)+\margin}:\innerfrac*\radius)  arc [start angle={360/\n * (7-1)+\margin}, end angle={360/ \n *(9-1)-\margin}, radius=\innerfrac*\radius] node [midway,above left]{$P_0$};

\end{tikzpicture}}
    
    \caption{Constructing a simple cycle by Algorithm~\ref{alg:ED_Cycle} and  extending an ear by Algorithm~\ref{alg:ED_Ear}.}
    \label{fig:robbins}
\end{figure}

\subsection{Formal description}
We now formally define our construction. Each node holds a variable named~$\Cycle$ that contains a global representation of the current~$C_i$. At the same time,
the simple cycle $root_i \xrightarrow[E_i]{} z_i \xRightarrow[C_i]{} root_i$ is represented locally, using the variables $\Next_v$ for the clockwise neighbor of~$v$ and $\Prev_v$ for its counterclockwise neighbor.

In our algorithms, the first ID in the variable~$\Cycle$ is the current root.
When the root node changes, each node~$v$ locally rotates the sequence of IDs in $\Cycle_v$ (say,  clockwise), so that the new root becomes the first ID in the string. 

The pseudo-code for our content-oblivious protocol for constructing a Robbins cycle appears in Algorithms~\ref{alg:ED_Cycle} and~\ref{alg:ED_Ear}. 
These use as sub-procedures the protocols~$\Pi_{\textsf{learnID}}$ and~$\Pi_{\textsf{NextRoot}}$, which are the content-oblivious versions of Algorithms \ref{alg:learn_ids} and~\ref{alg:next}, obtained by simulating them through Theorem~\ref{thm:main-general}.
Note that all these algorithms share the same variables, i.e., $\Cycle_v$, $\Prev_v$, and $\Next_v$ of node~$v$.

Our protocols use the ability to broadcast a message on a cycle defined either locally or globally. 
To be more accurate, the instruction ``broadcast M'' and ``wait for message M'' are to be understood as sending the message~$M$ with destination~$*$ and receiving any message with destination~$*$, respectively, using the method of Remark~\ref{rem:broadcast}. 
The sender also receives the broadcast message after all other nodes receive it and acts upon the pseudo-code for processing it. This guarantees synchronization, i.e., that the sender does not continue before all other nodes receive the broadcast message, which is crucial, for example, when we switch the underlying cycle we communicate over.
Indeed, in the noise-resilient protocol, the sender holds the token and does not release it before it gets the $\END$ pulse for that message, and by this time all other nodes receive that message as well. If now all nodes change their $\Cycle_v$, then the next pulse sent by the root goes through the new cycle.

\renewcommand\thealgorithm{\the\numexpr\value{algorithm}-2(a)}

\begin{algorithm}[htp]
\caption{Content-oblivious Ear-Decomposition: Closing an ear for the first time}\label{alg:ED_Cycle}
\begin{algorithmic}[1]

\State \textbf{Init:} Set $\Pi_{\textsf{learnID}}$ and $\Pi_{\textsf{NextRoot}}$ to be the content-oblivious versions of Algorithms~\ref{alg:learn_ids} and~\ref{alg:next}, respectively, obtained via Theorem~\ref{thm:main-general}.

\Statex
\Statex \textbf{node $v$, upon initialization}:
\State $\mathit{state}_v \gets \mathsf{init}$, $\Next_v\gets \bot$, $\Prev_v\gets \bot$, $\Cycle_v\gets\epsilon$. 
All edges unmarked. 
\If {$v$ is the root}
    \State choose an arbitrary edge~$(v,u)$ 
    \State send a pulse to~$u$ and mark the edge $(v,u)$ as \textsf{used}.
    \State $\Next_v \gets u$, $\mathit{state}_v \gets \mathsf{DFSroot}$ 
\EndIf

\Statex
\Statex \textbf{node $v$, upon receiving a pulse from~$w$}:
\If {$\mathit{state}_v = \mathsf{init}$}
	\State  $\Prev_v \gets w$, mark $(w,v)$ as \textsf{used} \label{line:set-in}
	\State choose an arbitrary neighbor~$u \ne w$ where $(v,u)$ is unmarked
	\State send a pulse to $u$ and mark  $(v,u)$ as \textsf{used}
    \State $\Next_v\gets u$, $\mathit{state}_v \gets \mathsf{DFS}$ \label{l:set_out}
\ElsIf {$\mathit{state}_v = \mathsf{DFS}$}  
    \If {$w = \Next_v$} \Comment{This is a cancellation pulse} \label{l:cancel_pulse_init}
    	\State choose an arbitrary neighbor~$u'$ where $(v,u')$ is unmarked:
    	\IndState send a pulse to $u'$, set $\Next_v \gets u'$ and mark $(v,u')$ as \textsf{used}
    	\label{line:DFS next edge}
    	
    	\If {no such $u'$ exists} 
    	    \State send a pulse to $\Prev_v$ \Comment{Send a cancellation pulse to parent}
    	    \State $\mathit{state}_v \gets \mathsf{init}$, $\Prev_v\gets \bot$, $\Next_v \gets \bot$, unmark all edges %
    	    \label{l:reset_after_cancel}
    	\EndIf

	\ElsIf {$w \ne \Prev_v$}  \Comment{A cycle is closed at $v$, but $v$ is not the root} \label{l:cancel_init}
		\State send a pulse to $w$
		and  mark $(v,w)$ as \textsf{used}.

	\Else ~($w=\Prev_v$)	\Comment{This is a second pulse---node is on a cycle}
		\State send a pulse to~$\Next_v$		%
		\State 
		$\Cycle_v\gets \Pi_{\textsf{learnID}}$, 
		executed  
		over the cycle locally defined by $\Prev_v,\Next_v$; initialize as non token holder.
		\label{line:EDcycle:learnID-node}
 		\label{line:Ei-learns-Ci}
		\State execute $\Pi_{\textsf{NextRoot}}$ over $\Cycle_v$; initialize as non token-holder.

	\EndIf
\ElsIf {$\mathit{state}_v = \mathsf{DFSroot}$} \label{line:ED_C_root}
	\State $\Prev_v \gets w$ 
	\label{line:EdCycle:root-closes-cycle}
	
	\State send a pulse to~$\Next_v$	\Comment{A cycle is closed, start communicating on it}
	
	\State wait until a pulse is received from $\Prev_v$ 
		\label{line:EDcycle:wait2ndPulse}
	\State $\Cycle_v \gets \Pi_{\text{learnID}}$, executed  over the simple cycle locally defined by $\Prev_v,\Next_v$; initialize as token holder.
	\label{line:EDcycle:learnID-root}
	
	\State execute $\Pi_{\textsf{NextRoot}}$ over $\Cycle_v$; initialize as token holder. 
	\label{l:start_iter}

\EndIf \label{line:EDCycle:END}

\end{algorithmic}
\end{algorithm}

\renewcommand\thealgorithm{\the\numexpr\value{algorithm}-3(b)}   %

\begin{algorithm}[htp]
\caption{Content-oblivious Ear-Decomposition: Ear extension}\label{alg:ED_Ear}
\begin{algorithmic}[1]

\setcounter{ALG@line}{\getrefnumber{line:EDCycle:END}}

\Statex \textbf{node $v$ marked as $root$, upon initialization}:
\State choose an edge $(v,u)\notin \Cycle_v$ and send a pulse to~$u$ 
\label{line:EDear:startDFS}
\State $\Next_v \gets u$

\Statex
\Statex \textbf{node $v$, upon receiving a pulse on $(v,u) \notin \Cycle_v$}: 
\State $\Prev_v \gets u$
\label{line:EDear:DFSends}
\State broadcast ``$\langle \mathsf{EarClosedAt}\rangle, v$'' over $\Cycle_v$

\Statex
\Statex \Comment{In parallel to the above, pulses from $\Cycle_v$ are interpreted as messages of a noise-resilient protocol}
\Statex \textbf{node $v$, upon receiving ``$\langle \mathsf{EarClosedAt}\rangle, w$'' on $\Cycle_v$}:
\State $P_i \gets $ the simple path $w \xRightarrow[\Cycle_v]{} root$  \label{l:define_pi} \Comment{$P_i=\emptyset$ if $w=root$}
\If {$v\in P_i$}
    \State set $\Prev_v$, $\Next_v$ according to~$P_i$  
    \label{line:EDear:Pi}
    \Comment{\smash{\parbox[t]{6cm}{The root sets $\Prev$ and $w$ sets $\Next$ (unless $root=w$); inner nodes set both}}}
\EndIf
\If {$v$ is the root} 
    \State send a pulse to~$\Next_v$
    \label{line:EDear:secnodPulse}
    \If {$root = w$}
        \State wait to receive a pulse from~$\Prev_v$ \Comment{A closed ear, the pulse will reach back the root}
        \State broadcast $\langle \mathsf{ready} \rangle$ on $\Cycle_v$
    \EndIf

\ElsIf {$v = w$} \Comment{ $w\ne root$}
    \State wait to receive a pulse from~$\Prev_v$ 
    \State broadcast $\langle \mathsf{ready} \rangle$ on $\Cycle_v$
    \label{line:EDear:sendREADY}
\EndIf
\State wait to receive $\langle \mathsf{ready} \rangle$ on $\Cycle_v$ 
 
\If {$\Prev_v,\Next_v \ne \bot$} \Comment{$v$ is on $P_i$}  	 \label{line:recAddEar}
        \State execute $\Pi_{\textsf{learnID}}$  over the simple cycle locally defined by $\Prev_v,\Next_v$; root is token holder. \label{line:Pi-learns-Ci}
        \State $\Prev_v\gets \bot$, $\Next_v \gets \bot$ 
        \label{line:EDear:afterLearnID}
\EndIf

\If {$v$ is the root}
    \State  broadcast ``$\langle \mathsf{NewCycle}\rangle, C_{i+1}$" over $\Cycle_v$, where $C_{i+1}$
        is the output of $\Pi_{\textsf{learnID}}$.
    	\label{line:AddEar}
\Else
    \State wait to receive the message ``$\langle \mathsf{NewCycle}\rangle, C_{i+1}$" over $\Cycle_v$.
\EndIf
\State $\Cycle_v \gets C_{i+1}$ \Comment{All nodes in $C_i$ switch to $C_{i+1}$; nodes on $E_i$ were set at line~\ref{line:Ei-learns-Ci}}
\State execute $\Pi_{\textsf{NextRoot}}$ over $\Cycle_v$; The root initializes as the token holder
\label{line:Edear:next}
	    
\end{algorithmic}
\end{algorithm}

\renewcommand\thealgorithm{\the\numexpr\value{algorithm}-3}

\begin{algorithm}[htp]
\caption{$\pi_{\textsf{learnID}}$, learning the IDs on a newly constructed ear (noiseless setting)}\label{alg:learn_ids}
\begin{algorithmic}[1]

\Statex \textbf{node $v$, upon initialization:}
\If{$v$ is the root}
    \State send $id(v)$ to $\Next_v$ 
\EndIf

\Statex
\Statex \textbf{node $v$, upon receiving} $m= (id_1,id_2,\ldots)$:

    \If { $id_1 \ne id(v)$}
    \Comment {$\{\Next_v\}_{v\in V}$ is guaranteed to induce a simple cycle}
        \State  $m'\gets m \| id(v)$
        \State send $m'$ to $\Next_v$

        \Else \Comment {Back to root, $m$ contains all the nodes on $\{\Next_v\}_{v\in V}$}

            \State $\mathit{new\_cycle}\gets \Cycle_v \| m$

            \State broadcast ``$\langle \text{\textsf{done}}\rangle,   \mathit{new\_cycle}$''

        \EndIf

\Statex
\Statex \textbf{node $v$, upon receiving} ``$\langle \text{\textsf{done}}\rangle,  C$'':
    \State return  $C$ 
\end{algorithmic}
\end{algorithm}

\begin{algorithm}[htp]
\caption{$\pi_{\textsf{NextRoot}}$,  choosing a new root (noiseless setting)}\label{alg:next}
\begin{algorithmic}[1]

\Statex \textbf{node $v$, upon initialization:}
\If{$v$ is the root}
    \State broadcast ``$\langle \textsf{check edges}\rangle$''
    \State wait to receive $|\{id(v') \mid v'\in \Cycle_v\}|$ many replies
    \If{received ``$\langle \textsf{has unexplored edges}\rangle$, $id(u)$''}
    \label{line:next:chooseNew}
    \Comment{Choose arbitrarily, if non unique}
        \State broadcast ``$\langle \textsf{new root}\rangle$, $id(u)$''
    \Else \Comment{All edges are explored}
        \State broadcast ``$\langle \textsf{completed}\rangle$''
    \EndIf
\EndIf

\Statex
\Statex \textbf{node $v$, upon receiving} $\langle \textsf{check edges}\rangle$:
\If{$v$ has unexplored edges} 

    \State broadcast ``\emph{$\langle \textsf{has unexplored edges}\rangle$, $id(v)$}''

\Else
\State broadcast ``$\langle \textsf{no unexplored edges}\rangle$, $id(v)$''

\EndIf

\Statex 
\Statex \textbf{node $v$, upon receiving }``$\langle \textsf{new root}\rangle$, $id(u)$'':
\Comment{Broadcast message is received also by its originator}

\State rotate $\Cycle_v$ clockwise until it starts with an occurrence of~$u$. The node $u$ is now marked root
\State execute Algorithm~\ref{alg:ED_Ear} \label{line:next:again}

\Statex
\Statex \textbf{node $v$, upon receiving} $\langle \textsf{completed}\rangle$:
\State terminate  \Comment{A Robbins cycle is constructed}
\label{alg:next:terminate}
\end{algorithmic}
\end{algorithm}

\subsection{Analysis}
Our main theorem in this section shows that \ALGfour constructs a Robbins cycle that includes all the edges in~$G$ despite a fully-defective environment. 
\begin{theorem}\label{thm:ED}
For any 2-edge-connected graph~$G$, 
\ALGfour
constructs a sequence of cycles $C_0,\ldots, C_k$, where $C_0$ is a simple cycle that includes the root, and $C_k$ is a Robbins cycle that contains all the edges~$E$ of~$G$.  
\end{theorem}

For the ease of the analysis, we define iterations of \ALGfour. %
We say that iteration $i+1$ begins 
when the $\Pi_{\textsf{NextRoot}}$ is being executed for the $i$-th time \emph{by a node which is currently marked as root}, i.e., when such a node 
reaches either Line~\ref{l:start_iter} or~\ref{line:Edear:next}.
Note that by the code, there can only be one root for each iteration.
We start with some helping lemmas.

\begin{lemma}\label{lem:ED_C0}
Suppose Algorithm~\ref{alg:ED_Cycle} is executed by all nodes in a 2-edge-connected graph~$G$, where a single node is marked as a root.
Then, the root node eventually 
reaches \Cref{line:EdCycle:root-closes-cycle},
and at that time,
there exists a single simple cycle~$C_0$, locally represented by the nodes on it. Furthermore, $root\in C_0$. 
\end{lemma}
\begin{proof}
It is immediate from the pseudo-code that Algorithm~\ref{alg:ED_Cycle} performs a sequential depth first traversal starting from the root and using marked edges to avoid repeating already visited edges.
We can think of the DFS as sending a DFS-token that progresses over non-visited edges until reaching a visited node~$v$. The DFS-token advances by sending a single pulse. 

Suppose the DFS-token reaches an already-visited node~$v$, this node is either the root, in which case we are done, 
or it is not the root. In the latter case,
the node $v$ sends the DFS-token back to where it came from, causing the DFS to backtrack that edge and continue with the DFS from the parent of~$v$ in the induced DFS tree.
Since the graph is 2-edge-connected, there exists a simple cycle that begins and ends at the root. 
A DFS search, once completed, explores all the edges in~$G$. Therefore, the DFS must eventually reach the root again and close a \emph{simple} cycle, defined by the progress of the DFS-token while ignoring any backtracked edges. Indeed, each node sets its $\Prev_v$ variable to the first node from which the DFS-token is received and sets its $\Next_v$ variable to be the node to which the DFS-token progresses. Backtracking an edge resets $\Prev_v,\Next_v$, accordingly in  Lines \ref{line:DFS next edge} or \ref{l:reset_after_cancel}.

Denote the above constructed cycle as~$C_0$.
We note that nodes that are not on $C_0$ are either never reached by the DFS or the DFS reaches them and backtracks since it does not reach the root from that path. In either case, their status at the time when the root reaches \Cref{line:EdCycle:root-closes-cycle}, and also at the end of Algorithm~\ref{alg:ED_Cycle}, is $init$ with no marked edges, and with $\Prev=\Next=\bot$. 
Therefore, $C_0$ is the only cycle  defined at this point. 
\end{proof}

Next, we observe that the nodes on $C_0$ switch to a global representation of their cycle.
\begin{lemma}
Once the root completes \Cref{line:EDcycle:learnID-root}, all the nodes on $C_0$ hold a global representation string of~$C_0$. 
\end{lemma}
\begin{proof}
Lemma~\ref{lem:ED_C0} establishes that once the root reaches \Cref{line:EdCycle:root-closes-cycle}, then $C_0$ is locally well-defined, i.e., every node that belongs to $C_0$ knows the previous and subsequent nodes in the cycle.
The root then sends a second pulse which progresses over $C_0$ and causes all the nodes on~$C_0$
to execute $\Pi_{\textsf{learnID}}$, where the root is the token holder (\Cref{line:EDcycle:learnID-root}) and other nodes are non token holders (\Cref{line:EDcycle:learnID-node}).
Note that the root awaits until the second pulse reaches it back (\Cref{line:EDcycle:wait2ndPulse}). 
By that time, all the other nodes on~$C_0$ start executing $\Pi_{\textsf{learnID}}$, but they are not token holders, so they remain idle. Only once the root starts executing $\Pi_{\textsf{learnID}}$, pulses are sent over~$C_0$ and the content-oblivious computation of Algorithm~\ref{alg:learn_ids} initiates.

The execution of Algorithm~\ref{alg:learn_ids} produces the sequence of IDs in $C_0$ according to the clockwise direction of the cycle: the root begins by sending its ID to its $\Next$ (clockwise) neighbor, which concatenates its ID, and so on. 
Once the message reaches the root again, it contains all the IDs of the nodes in~$C_0$ \emph{according to the clockwise direction of the cycle}. This string is then broadcast to all~$C_0$, so all the nodes now possess the global representation of~$C_0$ as required.
\end{proof}

Note that after the construction of~$C_0$ completes, the nodes that belong to~$C_0$ continue to execute Algorithm~\ref{alg:ED_Ear}, while the rest of the nodes are still executing Algorithm~\ref{alg:ED_Cycle}.
We now argue that the algorithm keeps adding edges to the currently-constructed cycle.

For a cycle~$C$, let us denote by $Edge(C)$ the set of edges in~$C$. 
We prove that each iteration of \ALGfour
constructs a larger cycle. That is, assuming the nodes on~$C$ execute Algorithm~\ref{alg:ED_Ear} while the rest of the nodes execute Algorithm~\ref{alg:ED_Cycle}, then at the end of that iteration, there is a globally defined cycle $C'$ such that all the nodes on $C'$ know this cycle (the other nodes keep executing Algorithm~\ref{alg:ED_Cycle}), and $C'$ is strictly larger than~$C$, that is,
$Edge(C) \subsetneq Edge(C')$.

\begin{lemma}\label{lem:nextCi}
Let $G$ be a 2-edge-connected graph and let $C_i$ be a cycle, such that $E\setminus Edge(C_i) \ne \emptyset$.
Let \emph{the root} be a single marked node on $C_i$ that is adjacent to an edge in~$E\setminus Edge(C_i)$.
Suppose nodes on $C_i$ all start executing Algorithm~\ref{alg:ED_Ear} while other nodes in~$G$ run Algorithm~\ref{alg:ED_Cycle} and their state is~$init$.
At the end of this iteration, 
there exists a
cycle~$C_{i+1}$ with $Edge(C_i) \subsetneq Edge(C_{i+1})$, all the nodes on $C_{i+1}$ know its global representation, and all the other nodes continue executing Algorithm~\ref{alg:ED_Cycle} and their state is~$init$. 
Further, if all the occurrences of any edge in~$Edge(C_i)$ have the same orientation, the same holds for~$C_{i+1}$.
\end{lemma}
\begin{proof}
Note that the nodes basically perform a DFS search over the unused edges, i.e., over all the edges except edges that belong to~$C_i$. The $root$ initiates the DFS search (\Cref{line:EDear:startDFS}). 
Since the root has at least one edge which does not belong to~$C_i$, denote the edge to which the root sends a pulse in \Cref{line:EDear:startDFS} by~$(root,v)$.

We argue that the DFS, after passing the DFS-token over $(root,v)$, must reach a node that belongs to~$C_i$ before it backtracks the edge~$(root,v)$. Suppose not, then there is no path between $v$ and any node in~$C_i$ that does not go through $(root,v)$. 
Hence, $(root,v)$ is a bridge, yet this is a contradiction since~$G$ is 2-edge-connected.

Once the DFS reaches some node~$z$ on~$C_i$ in \Cref{line:EDear:DFSends}, the path $E_i$ is well defined: it is the new ear---the path the token has taken from $root$ to~$z$, disregarding any backtracked edge. Note that $E_i$ is not empty and  $Edge(E_i)\subseteq E\setminus Edge(C_i)$, i.e.,  $E_i$ contains at least one new edge that does not belong to~$C_i$. 
Additionally, the path~$P_i$ constructed in \Cref{l:define_pi} is well defined: 
it is the shortest path between~$z$ and $root$ that uses only the directed edges in~$Edge(C_i)$. We know at least one such path exists since $z$ and $root$ are both nodes on the cycle~$C_i$, and take the lexicographic-first such path if multiple shortest-paths exist. Since all nodes on~$C_i$ know $Edge(C_i)$ then $P_i$ is agreed upon all of them.
Hence $C_{i+1} = C_i \| E_i \| P_i$ is a well defined cycle from $root$ to $root$ for which $Edge(C_i) \subsetneq Edge(C_{i+1})$.
It is easy to verify that all the occurrences of any edge in~$Edge(C_{i+1})$ have the same orientation: edges in~$E_i$ appear only once in~$C_{i}$, and all the other edges obey their orientation in~$C_i$, which is unique by assumption.

We now show that at the end of the iteration, all the nodes on~$C_{i+1}$ hold a global representation of~$C_{i+1}$ while the rest of the nodes remain in state $init$, executing Algorithm~\ref{alg:ED_Cycle}.
Note that as the DFS progresses through~$E_i$, all the nodes on~$E_i$ define their $\Next$ and $\Prev$ variables according to the progress of the DFS-token, so that the path $E_i$ is locally defined. 
After the DFS-token reaches $z$ in  \Cref{line:EDear:DFSends}, this node communicates over $C_i$ to let all the nodes of $C_i$ know that an ear is closed and its endpoints are $root$ and~$z$. 
With this information, each node on~$C_i$
can tell whether it belongs to~$P_i$, 
and if it is on~$P_i$, it can tell its successor and predecessor nodes on~$P_i$. 
Thus, each such node locally sets its $\Next$ and $\Prev$ variables according to the path~$P_i$ in \Cref{line:EDear:Pi}. Note that the concatenation of the two paths, $E_i \| P_i$, yields a simple cycle, 
locally defined by all the nodes on it. 
Also note that if $E_i$ is a closed ear, when $z=root$, then $P_i=\emptyset$, yet $E_i \| P_i$ is still a simple cycle.

Next, the root sends a second pulse in  \Cref{line:EDear:secnodPulse} which propagates along~$E_i$ and 
triggers the nodes on~$E_i$, except for $root$ and $z$, to start executing $\Pi_{\textsf{learnID}}$ on the cycle locally defined by their $\Next$ and $\Prev$ variables (\Cref{line:EDcycle:learnID-node}). However, none of the (inner) nodes on~$E_i$ is the token holders in the execution of~$\Pi_{\textsf{learnID}}$, so they remain idle, in the sense that they do not request the token.

Once this second pulse reaches~$z$ in  \Cref{line:EDear:sendREADY}, it informs the nodes in~$C_i$ about this event by broadcasting $\langle \mathsf{ready}\rangle$ on~$C_i$.  
Note that at this point, the nodes on~$C_i$ are all idle. Specifically, no node wishes to obtain the token, so no pulses are being sent over $C_i$. It is safe to switch to communicating over the locally defined simple cycle $E_i \| P_i$. 
The nodes on that cycle now execute $\Pi_{\textsf{learnID}}$, after which all of them learn the global string representing~$C_{i+1}=C_i \| E_i \| P_i$. 
At this point, the nodes in $E_i$ except $root$ and $z$ switch to communicate over $C_{i+1}$. However, they are not the token holders so they keep being idle until the rest of the nodes switch to~$C_{i+1}$, without interfering with them.

After $\Pi_{\textsf{learnID}}$ terminates, all the nodes on $E_i \| P_i$ that were executing it know it has terminated. The root is the last to obtain the final message ``$\langle \mathsf{done} \rangle, C_{i+1}$'', so at the time when the root finishes 
$\Pi_{\textsf{learnID}}$, 
all other nodes on $C_i$ are set to communicate over $C_i$: the nodes on $P_i$ are done with $\Pi_{\textsf{learnID}}$, and set $\Next=\Prev=\bot$ in  \Cref{line:EDear:afterLearnID}, and now await the $\langle \mathsf{NewCycle}\rangle$ message \emph{on~$C_i$}. The rest of the nodes on~$C_i$ do not perform the \textbf{if} statement of \Cref{line:recAddEar} and thus are already awaiting the $\langle \mathsf{NewCycle}\rangle$ message.

Finally, the root broadcasts ``$\langle \mathsf{NewCycle}\rangle, C_{i+1}$'' over~$C_i$ which causes all the nodes in~$C_i$ to change their $\Cycle$ variable to~$C_{i+1}$. 
The root is the last to finish the procedure of the broadcast invocation, 
and by that time, all nodes of~$C_{i+1}$ are set to the cycle~$C_{i+1}$ and idle. The root is the token holder and is expected to send the next message on~$C_{i+1}$.
\end{proof}

The proof of Theorem~\ref{thm:ED} can now easily be obtained as a corollary of the above lemma. Multiple invocations of Algorithm~\ref{alg:ED_Ear} eventually yield a Robbins cycle $C_k$ with $Edge(C_k)=E$. 

\begin{proof}[Proof of Theorem~\ref{thm:ED}]
By Lemma~\ref{lem:ED_C0}, we know that after the first iteration of Algorithm~\ref{alg:ED_Cycle} we obtain a simple cycle~$C_0$. If $C_0$ consists of all the edges of~$G$, we are done---the nodes run $\Pi_{\textsf{NextRoot}}$ to find out that all edges are exhausted, and the algorithm terminates in \Cref{alg:next:terminate} of Algorithm~\ref{alg:next}.
Otherwise, we keep executing Algorithm~\ref{alg:ED_Ear} with a new root that has an adjacent unused edge. This is done by Algorithm~\ref{alg:next}: each node broadcasts whether or not it has unused edges adjacent to it, along with its ID. The current root arbitrarily picks one node with unused edges (\Cref{line:next:chooseNew}) and broadcasts this choice to all the nodes of~$C_i$. Since all the nodes possess a global representation of $C_i$, they can rotate it so that the new root becomes first in the global representation, which is consistent among all nodes and allows, for example, to determine $P_i$ in a consistent manner.
Then, Algorithm~\ref{alg:ED_Ear} is invoked again with this chosen node as the new root (\Cref{line:next:again}).
At this point, the statement of Lemma~\ref{lem:nextCi} holds: there is a cycle $C_i$ globally represented by all the nodes in it, there is a single root on~$C_i$ and it has adjacent unused edges, and all the nodes in $G \setminus C_i$ are in state $init$ in the execution of Algorithm~\ref{alg:ED_Cycle}.

By Lemma~\ref{lem:nextCi}, every iteration of the algorithm starting on~$C_i$ produces a cycle $C_{i+1}$ with at least one additional edge in~$E$ that does not appear in~$C_i$. 
It is easy to verify that, as long as some edge is still unused, at the end of constructing~$C_{i+1}$, i.e., after executing \Cref{line:Edear:next} but before the nodes re-iterate Algorithm~\ref{alg:ED_Ear} (\Cref{line:next:again} of Algorithm~\ref{alg:next}), the requirements for Lemma~\ref{lem:nextCi} hold with respect to the newly constructed cycle.
Thus, after at most $|E|-|Edge(C_0)|$ iterations of Algorithm~\ref{alg:ED_Ear}, the obtained cycle consists of all the edges~$E$ in~$G$. Since each edge has a single  orientation induced by the cycle (this clearly holds for the simple cycle~$C_0$, and inductively throughout the construction), and since all the nodes in~$G$ appear in the obtained cycle, it is a Robbins cycle.
\end{proof}

\begin{remark}
In order to communicate   over any intermediate (non-simple) cycle $C_i$ via \ALGthree, 
a single node-occurrence must be defined as the token holder. Furthermore, all other nodes must know the segment in~$C_i$ that contains that designated node-occurrence. Recall that in \ALGthree, each node maintains the invariant that the token resides in its segment~0 (see Section~\ref{sec:GeneralCycle}).
Our construction indeed provides the nodes with this information, 
which can be retrieved from the  global representation of~$C_i$. 
The first node-occurrence in~$C_i$ is defined to be the token holder, and each other node can re-number its occurrences along~$C_i$ in the natural manner, so it is consistent with having the token at its segment~0.
The above also holds also for the Robbins cycle~$C_k$ constructed in \Cref{thm:ED}. 
\end{remark}

\begin{remark}
\label{rem:localCycle}
\textbf{Avoiding Global Knowledge:}
In the above construction, the nodes obtain a global representation of the cycles $C_i$ they construct. We remark that this knowledge helps in simplifying the construction and reducing the length of the constructed cycle. 
However, it is not necessary, and a similar construction can be designed in which each node only holds local information about~$C_i$, i.e., only its clockwise and counterclockwise neighbors for each of its occurrences on~$C_i$. 
We provide here the main differences in such a construction.

\textbf{(1)} 
The global representation of $C_i$ is used to determine the path $P_i$ between the end points $(root,z)$ of the newly constructed ear~$E_i$. For the above construction to work, we need every node to know whether or not it belongs to~$P_i$; if it is part of~$P_i$, then it should appear one more time in $C_{i+1}$.
Now, suppose that every node~$v$ on~$C_i$  knows only a local representation of~$C_i$, namely, its $\Next$ and $\Prev$ neighbors for each occurrence of~$v$ on~$C_i$. 
The path $P_i$ can be determined in the following way. Once the endpoint~$z$ of the ear~$E_i$ broadcasts the message ``$\langle \mathsf{EarClosedAt}\rangle,z$'' over $C_i$, all the nodes in $C_i$ switch to a new state of ``detecting~$P_i$''. In this state, if a node-occurrence receives a clockwise pulse, it means that this occurrence belongs to~$P_i$. 
A counterclockwise pulse  signifies that the node-occurrence should quit this new state and continue executing \Cref{alg:ED_Ear}. In both cases, each pulse is propagated by the node-occurrence along the same direction it is received.

The nodes use the above mechanism as follows. Once the broadcast of ``$\langle \mathsf{EarClosedAt}\rangle,z$'' completes at~$z$, it sends a single clockwise pulse. This pulse propagates along~$C_i$ until it reaches a node-occurrence of the root; denote by~$P_i$ the path that this pulse has taken. The root \emph{does not} propagate the pulse, but instead sends a single counterclockwise pulse, which travels along the entire~$C_i$ until reaching that same root node-occurrence again. 
At this point, all the node-occurrences that belong to~$P_i$ have received a clockwise pulse, and all the node-occurrences on~$C_i$ have received a counterclockwise pulse, so all nodes can continue with the construction as above.
Note that this method also allows the nodes to track the segment in which the root lies, so that at the end of the construction they can infer the token segment at any step. 

\textbf{(2)}
The other place our construction uses the global representation is in $\pi_{\textsf{NextRoot}}$, where the root awaits to receive a message from every node on~$C_i$ to know whether the construction is done. However, without a global representation, the root does not know how many nodes are in~$C_i$ and thus it cannot know how many messages to expect. 
The remedy for this issue utilizes the token delivery method of \ALGthree. Namely, we replace \Cref{alg:next} with the following method. The root begins by broadcasting $\langle \textsf{check edges}\rangle$. 
Every node that still has an unexplored edge requests the token, and if it receives the token, it sends its ID. The first node to do so becomes the new root. 
If no such node exists, the token  propagates until it reaches the (old) root again. In this case, the root acquires the token and broadcasts $\langle \textsf{completed}\rangle$ to indicate that the Robbins construction is done.
\end{remark}

\begin{remark}\textbf{Coping with $KT_0$:}
\label{rem:KT0}
Algorithm~\ref{alg:learn_ids} and its noise-resilient form~$\Pi_{\textsf{learnID}}$ are $KT_1$ algorithms, in which each node knows the IDs of its neighbors. 
We remark that we can establish the learn-ID functionality, and thus the construction of the Robbins cycle, even in $KT_0$ networks, in which the IDs of the neighbors of a node are not known to it upon initialization. Note that Algorithm~\ref{alg:learn_ids} as stated cannot work in a $KT_0$ network since a node does not know which node comes immediately next to it in the cycle. In other words, after the root sends its ID as the first message, this message reaches \emph{all} other nodes and none of them knows they are the next one on~$C_0$.

We can solve this issue by relying on the order in which the token holder shifts in the underlying simulator. 
A $KT_0$ protocol for learning the IDs starts by instructing all the nodes to broadcast their ID. Thus, all nodes request to be token holders. Once the root sends its own ID and releases the token, its immediate \emph{counterclockwise} neighbor becomes the new token holder. Thus, the IDs are broadcast exactly in their counterclockwise order on~$C_0$. Once the root becomes a token holder again, this process is done. 

We also note that the simulator of \Cref{sec:GeneralCycle} only requires local knowledge of a Robbins cycle  and thus can run on $KT_0$~networks with the above pre-processing step.
Thus, Theorem~\ref{thm:main} holds for $KT_0$~networks as well.
\end{remark}

\subsection{The length of the obtained Robbins cycle}
\label{sec:robbins-cycle-length}
We complete this section with a crude analysis of the size of  Robbins cycle our construction obtains 
and the communication complexity of the construction.

\begin{lemma}
Let $G$ be a 2-edge-connected graph, and let $C$ be the Robbins cycle constructed by Theorem~\ref{thm:ED}. Then $|C|=O(n^3)$. Further, \ALGfour communicates $O(n^8\log n )$ pulses altogether.
\end{lemma}
\begin{proof}
Given some~$C_i$, it holds that $|C_{i+1}|= |C_i| + |E_i|+ |P_i|$. Since $P_i$ is a shortest (simple) path between two nodes, we have $|P_i| < n $, for all iterations~$i$. 
A bound on the worst-case length of the Robbins cycle is obtained by considering $O(n^2)$ iterations of \ALGfour, in each of which, adding only a single edge to the current~$C_i$. 
In this case, the cycle's length extends by~$O(n)$ in each of the~$O(n^2)$ iterations, yielding a total length of~$O(n^3)$.

Let us now bound the communication complexity. 
Consider the iteration where the nodes begin with~$C_i$ and construct~$C_{i+1}$.
The $\pi_{\mathsf{learnID}}$ algorithm communicates at most $\alpha_i=|E_{i}|+|P_i|$ messages, each of length at most~$O(\alpha_i\log n)$, except for the $\langle \mathsf{done} \rangle$ message whose length is~$O(|C_{i+1}|\log n)$.
The $\pi_{\mathsf{NextRoot}}$ algorithm communicates $|C_{i+1}|$ messages of length $O(\log n)$.
The rest of \Cref{alg:ED_Ear} makes $O(1)$ broadcasts of messages of length $O(\log n)$, and a single 
$\langle \mathsf{NewCycle} \rangle$ message whose length is $\alpha_i$.
Recall that by \Cref{lemma:g-cost-binary},
broadcasting a message of length $m$ over the cycle~$C_i$ takes 
$O( |C_i|(m+\log n))$ pulses.

Next, we argue that the DFS search within a single iteration of \ALGfour sends $O(n^2)$~pulses. To see that, recall that each edge is marked as used once the DFS-token passes through it. Additionally, the token might backtrack that edge, but no more pulses should be sent on that edge, leading to a total of at most $2|E|=O(n^2)$~pulses overall. 
The above does not hold for nodes that have backtracked all their edges and reset their state to~$init$, because they also unmark all their edges and might re-send pulses over edges that were already explored in this iteration.
We argue, however, that such nodes will never get the DFS-token again during that iteration. 
Indeed, assume towards contradiction that $u$ is a node that has reset its state during the current iteration and is \emph{the first} node that receives the DFS-token after resetting its state, say, over the edge $(u,v)$. 
Since $u$ has explored and backtracked all its edges, the DFS-token must have already passed through the edge~$(u,v)$ previously in this iteration. 
Therefore, it is marked $\mathsf{used}$ by~$v$, and it is impossible that $v$ sends a DFS-token over this edge, unless $v$ resets its state and unmarks all its edges. 
However, if $v$ reset its state and then sends a DFS-token over $(u,v)$, then $v$ must have received the DFS-token after resetting and before~$u$ did, contradicting our choice of~$u$.

We then conclude that the complexity of constructing the Robbins cycle in \ALGfour is bounded by
\[
\sum_{i}
\left [
\alpha_i \cdot O(\alpha_i \cdot \alpha_i\log n) + O(\alpha_i \cdot |C_i|\log n)
+ |C_{i+1}|\cdot O(|C_{i+1}|\log n)
+ O(|C_i| \log n)
+ O(n^2)
\right]
\]
pulses. Bounding $\alpha_i = O(n)$ and  $|C_i|,|C_{i+1}|=O(n^3)$, and the number of iterations $i \le |E| = O(n^2)$, we conclude that the complexity of constructing the Robbins cycle 
is $O(n^8\log n)$ pulses. 
\end{proof}
Note that the complexity can be reduced if we assume $KT_1$ networks and global representation of the constructed cycle.
Instead of terminating when all the adjacent edges of all the nodes were explored, we terminate when all nodes see that all their neighbors appear on the current~$C_i$. 
Each node can determine this information assuming $KT_1$ knowledge and a global representation of the cycle. 
This guarantees that at least one node is added at each iteration of \ALGfour, which reduces the number of iterations to~$i\le n$. This method leads to a Robbins cycle of total length~$O(n^2)$ and a communication complexity of $O(n^6\log n)$.

\section{Impossibility of resilient communication in fully-defective networks which are not 2-edge connected}
\label{sec:impossibility}

In this section we complement our simulator for 2-edge-connected graphs, with a proof showing that 2-edge connectivity is required for communication in fully-defective networks. The intuitive argument is that if the communication network is not 2-edge connected, then a bridge exists, and corrupting messages over that edge will lead to disconnecting the network, preventing the correct computation of any non-trivial function. Towards that goal we show the impossibility of asynchronous computation with \emph{two parties} in the presence of fully-defective channel noise. The two-party impossibility implies a general impossibility result for any network that contains a bridge since the two connected components over the two sides of the bridge can be reduced to the two parties case.

Formalizing the above intuition is slightly more subtle. For the impossibility to hold, we must require the protocol to give output (or explicitly terminate). To see why, 
consider the case of two parties (say, Alice and Bob) that hold the private inputs $x$ and~$y$, respectively, and need to compute some fixed known function~$f(x,y)$. 
Suppose that, instead of requiring the protocol to give a non-revocable output, we only require that there exists a time~$t$ after which both parties hold $f(x,y)$ and never change it again. 
Then, the following protocol succeeds in computing $f$ in the fully-defective two-party network
(stated for Alice; Bob's protocol is symmetric):
(a) Send  $x$ messages to Bob; (b) $\mathsf{count\gets 0}$; (c) Upon the reception of a message, $\mathsf{count \gets count + 1}$; update the output variable to $f(x,\mathsf{count})$.

\smallskip
Nevertheless, if we require the parties to terminate or to give an output,  no protocol for non-trivial functions~$f$ exists.

\begin{theorem}\label{thm:bridge-impossibility}
Consider a fully-defective network of two parties connected via a single noisy channel, and let $f(x,y)$ be any non-constant function.
Any two-party deterministic protocol that computes~$f$ and gives an output, is incorrect.
\end{theorem}

\begin{proof}
Let $f$ be some non-constant function and assume, without loss of generality, that its input and output domains are the natural numbers.
We can restrict the discussion to protocols in which each message sent by any of the parties contains a single `1' bit. This is without loss of generality, since we can equivalently consider the case where the adversary  corrupts the content of any message to be `1'.
Since the setting is asynchronous, a party can send zero or more messages as a function of its input and the number of messages it has received so far. A party is assumed to be idle between the time it sends a batch of messages until the time a new message arrives (which may trigger the transmission of new messages). In particular, once a new message arrives, the party immediately decides upon the number $k\ge 0$ of new messages to send, transmits them, and then goes back to being idle (or terminates).

Consider some inputs
$(x,y)$ and $(x',y)$ for which $f(x,y)\ne f(x',y)$, if no such inputs exist then a symmetric proof holds for a pair of inputs $(x,y)$ and $(x,y')$.
Fix Bob's input to~$y$. Note that once $y$ is fixed, Bob's actions  depend only on the number of messages he has received so far.
That is, we can completely describe Bob's protocol by the sequence %
$\mathcal{B}_y=(0, \mathsf{action_0}) (1, \mathsf{action_1})(2, \mathsf{action_2})\cdots$, where for any $t\ge0$, the item $(t, \mathsf{action_t})$ is to be interpreted as the action Bob performs after seeing 
$t$~messages from Alice. The value $\mathsf{action_t} \in \{ \mathsf{send}_k, \mathsf{SendAndOutput}_{k,r}\}_{k,r\ge0} $ describes the action Bob takes at that step of the protocol: $\mathsf{send}_k$ means that Bob transmits $k$~messages to Alice, and $\mathsf{SendAndOutput}_{k,r}$ means that Bob sends $k$ messages to Alice and sets its output register (irrevocably) to~$r$, i.e., Bob commits to the output $r$. Note that this is a complete characterization of Bob's protocol. 
We may assume that Bob continues to send and receive messages after setting its output, however, if in a later step Bob performs the action $\mathsf{SendAndOutput}_{k,r}$, then Bob will only send $k$ messages but the output register will not change.

Also note that Bob progresses sequentially. That is, Bob first performs $\mathsf{action_0}$, then $\mathsf{action_1}$, etc. 
Once Bob receives no further messages from Alice, he stops making any further progress. Thus, in order to give an  output, Bob must reach some $t\ge0$ where $\mathsf{action}_{t} = \mathsf{SendAndOutput}_{k,r}$.
Consider~$\mathcal{B}_y$ and set $\hat t=\argmin_t (\mathsf{action}_t \in \{\mathsf{SendAndOutput}_{k,r}\}_{k,r\ge0})$; we know that $\hat t<\infty$ and 
$\mathsf{action}_{\hat t}=\mathsf{SendAndOutput}_{\hat k,\hat r}$, with some $\hat k,\hat r\ge0$, or otherwise Bob never gives an output on input~$y$.
Finally, we note that Bob acts as described regardless of Alice's input: Bob advances sequentially until seeing $\hat t$ messages from Alice, after which it commits on the output~$\hat r$. 

Now consider an execution of the protocol on the input $(x,y)$. As described above, Bob commits on output when performing $\mathsf{action}_{\hat t}=\mathsf{SendAndOutput}_{\hat k,\hat r}$. If Bob does not give the correct output, we are done. 
Otherwise, $\hat r=f(x,y)$.
Next, consider the execution of the protocol on the input~$(x',y)$. If Bob receives less than $\hat t$ messages overall (and the protocol then reaches quiescence), Bob does not give an output. Otherwise, upon receiving the $\hat t$-th message, Bob outputs $\hat r=f(x,y)$. As both these options are incorrect for the input~$(x',y)$, we have reached a contradiction.
\end{proof}

\begin{acks}
This project has received funding from the European Union’s Horizon 2020 research and innovation programme under grant agreement no. 755839.
Ran Gelles is supported in part by the Israel Science Foundation (ISF) through Grant No.\@ 1078/17 and the United States-Israel Binational Science Foundation (BSF) through Grant No.\@ 2020277.
Gal Sela is supported in part by the Israel Science Foundation (ISF) through Grant No.\@ 1102/21.
\end{acks}

\newcommand{\etalchar}[1]{$^{#1}$}

\end{document}